\documentclass{article}
\usepackage{graphicx} 
\usepackage{amsmath}
\usepackage{abstract}
\usepackage{amsthm}
\usepackage{mathtools}
\usepackage{amssymb}
\usepackage{breqn}
\usepackage{enumitem}
\usepackage{subcaption}
\usepackage{float}
\usepackage{caption}    
\usepackage[dvipsnames]{xcolor}
\usepackage{geometry}
\usepackage{appendix}
\usepackage{titlesec}
\usepackage{epsfig,color}
\usepackage{hyperref}
\newtheorem{theorem}{Theorem}[section]
\newtheorem{proposition}{Proposition}[section]
\newtheorem{lemma}{Lemma}[section]
\newtheorem{remark}{Remark}[section]
\numberwithin{figure}{section}
\numberwithin{table}{section}
\numberwithin{equation}{section}

\theoremstyle{definition}
\newtheorem{example}[theorem]{Example}

\title{Ground states and droplet regimes of the extended Gross–Pitaevskii equation with Lee–Huang–Yang correction}
\author{
Weijie Huang \thanks{School of Mathematics and Statistics, Beijing Jiaotong University, Beijing 100044, People's Republic of China {\tt (wjhuang@bjtu.edu.cn)}}, \
Yang Liu \thanks{School of Mathematics and Statistics, Beijing Jiaotong University, Beijing 100044, People's Republic of China {\tt (yangliu02@bjtu.edu.cn)}}, \ and
Xinran Ruan \thanks{(Corresponding author) School of Mathematical Sciences, Capital Normal University, Beijing 100048, People’s Republic of China {\tt (xinran.ruan@cnu.edu.cn)}.}
}
\date{}

\begin{document}
\maketitle
\begin{abstract}
We study the ground states of the extended Gross--Pitaevskii equation with the Lee--Huang--Yang correction from both theoretical and numerical perspectives. Starting from the three-dimensional model, we derive reduced one- and two-dimensional equations through nondimensionalization and dimensional reduction. We establish existence and nonexistence results for ground states in different spatial dimensions, both in free space and under confining external potentials. For the numerical computation of ground states, we propose a normalized gradient flow method with a Lagrange multiplier. 
The numerical results show how the model parameters affect the ground-state profiles, and reveal different regimes in the free-space parameter plane, including no-ground-state, soliton-like, and droplet-like regions. We also introduce a simple flat-top approximation for the droplet regime and present two- and three-dimensional computations to illustrate more general localized structures.
\end{abstract}


\section{Introduction}
    
Quantum droplets in ultracold Bose gases have attracted much attention
and are of great interest in the study of many-body effects
and superfluid dynamics.
Unlike classical droplets stabilized by surface tension, Bose quantum droplets
originate from the balance between attractive mean-field interactions
and repulsive quantum fluctuations beyond the mean-field description \cite{2015Quantum}.
A classical description of this effect was given by T. D. Lee, K. Huang, and C. N. Yang in 1957 through the pseudopotential approach, leading to the well-known Lee--Huang--Yang correction \cite{1957Eigenvalues}.
The Gross--Pitaevskii equation is a standard mean-field model for dilute Bose gases \cite{1961Structure,1961Vortex}.
Incorporating the Lee--Huang--Yang correction into the Gross--Pitaevskii framework gives rise to the extended Gross--Pitaevskii model, which has become a basic model for describing quantum droplets and related phenomena.
In this context, Petrov predicted in 2015 the existence of self-bound quantum droplets in Bose--Bose mixtures \cite{2015Quantum}, and such states were subsequently observed in several experiments \cite{cabrera2018quantum,ferrier2016observation,schmitt2016self,semeghini2018self}.
These developments make the extended Gross--Pitaevskii model a natural framework for investigating the static and dynamical properties of quantum droplets.

Accordingly, we consider the following extended Gross--Pitaevskii (eGP) equation with the Lee--Huang--Yang (LHY) correction
\begin{equation}
    i\hbar\partial_t \psi(\mathbf{x}, t)
    = \left[ -\frac{\hbar^2}{2m}\nabla^2 + V(\mathbf{x})
    + g|\psi|^2 + g_{\text{LHY}}|\psi|^{3}\right] \psi,
    \label{eGPE_dim}
\end{equation}
where \(\psi\) denotes the macroscopic wave function,
\(\hbar\) is the reduced Planck constant,
\(m>0\) is the atomic mass,
\(V(\mathbf{x})\) is the external potential,
\(g = \frac{4\pi\hbar^2 a_s}{m}\) is the interaction coefficient,
\(g_{\text{LHY}} = C_{\text{L}} g a_s^{3/2}\) is the LHY coefficient,
\(a_s\) is the scattering length,
and \(C_{\text{L}}\) is a positive constant.
The wave function satisfies the normalization condition
\begin{equation}
	\|\psi\|_2^2 := 
    \int_{\mathbb{R}^3} |\psi|^2 \, d\mathbf{x} = N,
    \label{normalization}
\end{equation}
which corresponds to the conservation of  the total particle number \(N\).
The energy functional associated with \eqref{eGPE_dim} is given by
\begin{equation}
    E(\psi) = \int_{\mathbb{R}^3}
    \left[
    \frac{\hbar^2}{2m}|\nabla \psi|^2
    + V(\mathbf{x})|\psi|^2
    + \frac{1}{2}g|\psi|^4
    + \frac{2}{5}g_{\text{LHY}}|\psi|^5
    \right] d\mathbf{x}.
    \label{energy_dim}
\end{equation}
The ground state wave function \(\phi_g\) is defined as a minimizer of the energy functional \eqref{energy_dim} under the normalization constraint \eqref{normalization}, namely,
\[
\phi_g \in \arg\min_{\|\phi\|_2^2 = N} E(\phi).
\]


The existence and qualitative properties of ground states
for eGP-type models are fundamental issues in mathematical analysis,
and they also provide the theoretical basis for the design and
justification of numerical methods.
In recent years, increasing attention has been paid to these questions.
Y. Luo and A. Stylianou studied standing waves for a three-dimensional dipolar model with quantum fluctuations and three-body interactions \cite{2021On}.
They also obtained related results on ground states for a nonlocal cubic--quartic Gross--Pitaevskii equation in \cite{2018Ground}.
In a related direction, B. Feng, L. Cao, and J. Liu proved the existence of stable standing waves for the LHY-corrected dipolar Gross--Pitaevskii equation with partial harmonic confinement \cite{feng2021existence}.


At present, the numerical approximation of the eGP model with the LHY correction
has received only limited attention.
For the computation of ground states in classical Bose--Einstein condensates,
a variety of effective numerical methods have been developed.
Among them, the gradient flow with discrete normalization proposed by W. Bao and Q. Du
is one of the most widely used approaches \cite{bao2004computing}, and a
systematic review can be found in \cite{2013Mathematical}.
This method has been extended in several directions, including
Sobolev gradient methods \cite{danaila2010new}, nonlinear conjugate gradient
methods \cite{antoine2017efficient}, and
energy-stable gradient flow methods based on Lagrange multipliers \cite{liu2021normalized} or SAV-type
reformulations \cite{zhuang2019efficient}.
For multi-component and spinor Bose--Einstein condensates, related extensions
have also been developed \cite{bao2008computing,2005Gauss}.
Besides gradient-flow-based approaches, direct energy minimization methods have
also been proposed, including finite element methods for computing dark and
bright solitons \cite{bao2013numerical}, Riemannian optimization \cite{danaila2017computation} and regularized
Newton-type methods \cite{wu2017regularized}.

These methods provide useful starting points for the numerical treatment of the
eGP model.
However, the presence of the LHY correction introduces additional difficulties.
On one hand, the higher-order nonlinear term increases the nonlinearity of the model and introduces additional difficulties for numerical stability.
On the other hand, in the absence of an external potential,
that is when \(V(\mathbf{x}) \equiv 0\), self-bound quantum droplets are typically strongly
localized, which may lead to insufficient resolution, spurious oscillations,
or reduced accuracy in computations.
Therefore, classical methods  need to be further adapted
and improved in order to achieve a better balance among stability, accuracy,
and computational efficiency for the eGP model with the LHY correction.

 In this paper, we study the ground states of the eGP equation with the LHY correction
from both theoretical and numerical perspectives.
We first establish the existence and nonexistence results
for ground states in different spatial dimensions and parameter regimes.
For the numerical computation of ground states, we develop a GFDN-based method with a new discretization strategy, and numerical experiments show that the proposed scheme allows for a wider range of discretization parameters while maintaining
good stability and accuracy.
Our numerical results illustrate how the model parameters affect the ground-state profiles and identify several typical regimes of behavior.
We also present numerical results for more general settings, including cases
with external potentials and higher-dimensional configurations. 
In the present work, we restrict ourselves to the single-component model,
while dipolar and two-component cases will be left for future study.

The rest of this paper is organized as follows.
Section \ref{sec:nondim} is devoted to the nondimensionalization of the model and the derivation of a unified form of the eGP models in different spatial dimensions.
In Section \ref{sec:theory}, we investigate the existence and nonexistence of ground states with and without external potentials.
Section \ref{sec:scheme} presents a numerical scheme for ground-state computation.
In Section \ref{sec:experiment}, numerical results are provided for cases with and without external potentials.
Finally, some concluding remarks are presented in Section \ref{sec:conclusion}.

\section{Dimension Reduction} \label{sec:nondim}

We start from the three-dimensional eGP equation \eqref{eGPE_dim}.
To facilitate the subsequent derivation of reduced lower-dimensional models, we first rewrite the equation
in dimensionless form under the normalization constraint \eqref{normalization}.
For a prescribed constant \(c>0\), we introduce the scaling
\begin{equation}
    \tilde{t} = \frac{t}{t_s},\qquad
    \tilde{\mathbf{x}} = \frac{\mathbf{x}}{x_s},\qquad
    \tilde{\psi}(\tilde{\mathbf{x}},\tilde{t})
    = \frac{c x_s^{3/2}}{\sqrt{N}}\,\psi(\mathbf{x},t),
    \label{scaling}
\end{equation}
with
\(
    t_s = \frac{m x_s^2}{\hbar},
\)
so that the coefficient in front of the kinetic term becomes \(1/2\).
Then, omitting the tildes for notational simplicity and still denoting the rescaled potential by \(V\), we obtain the dimensionless three-dimensional eGP equation
\begin{equation}
    i\partial_t \psi(\mathbf{x},t)
    = -\frac{1}{2}\nabla^2\psi
    + V(\mathbf{x})\psi
    + \beta |\psi|^2\psi
    + \lambda |\psi|^3\psi,
    \label{eGPE_nodim}
\end{equation}
where the wave function satisfies the dimensionless normalization condition
\begin{equation}
    \|\psi(\cdot,t)\|_2 = c,
    \qquad \text{or equivalently,} \qquad
    \int_{\mathbb{R}^3} |\psi(\mathbf{x},t)|^2\,d\mathbf{x} = c^2.
    \label{normalization_nodim}
\end{equation}
The dimensionless interaction parameters are given by
\begin{equation}
    \beta = \frac{m g N}{\hbar^2 c^2 x_s}
          = \frac{4\pi N a_s}{c^2 x_s},
    \qquad
    \lambda = \frac{m g_{\mathrm{LHY}} N^{3/2}}{\hbar^2 c^3 x_s^{5/2}}
            = 4\pi C_{\mathrm L}\frac{N^{3/2}}{c^3}
              \left(\frac{a_s}{x_s}\right)^{5/2}.
\end{equation}
The associated dimensionless energy functional is
\begin{equation}
    E(\psi) = \int_{\mathbb{R}^3}
    \left[
    \frac{1}{2}|\nabla\psi|^2
    + V(\mathbf{x})|\psi|^2
    + \frac{\beta}{2}|\psi|^4
    + \frac{2\lambda}{5}|\psi|^5
    \right] d\mathbf{x}.
    \label{energy_nodim}
\end{equation}

As an important example, we consider the harmonic oscillator potential,
\begin{equation}
    V(\mathbf{x}) = V_{\mathrm{ho}}(x) + V_{\mathrm{ho}}(y) + V_{\mathrm{ho}}(z),
    \qquad
    V_{\mathrm{ho}}(\alpha) = \frac{1}{2}m\omega_\alpha^2\alpha^2,
    \qquad \alpha = x,y,z,
    \label{harmonic}
\end{equation}
for which the corresponding dimensionless potential takes the form
\begin{equation} \label{V_har}
    V(\mathbf{x})
    = \frac{1}{2}\bigl(\gamma_x^2 x^2 + \gamma_y^2 y^2 + \gamma_z^2 z^2\bigr),
\end{equation}
where
\begin{equation}
    \gamma_\alpha = \frac{m x_s^2}{\hbar}\omega_\alpha,
    \qquad \alpha = x,y,z.
\end{equation}

\begin{remark}
We note that \(\omega_\alpha\) may also vanish, corresponding to the case
without an external potential.
This differs from the conventional Bose-Einstein condensate (BEC) setting, where an external trapping
potential is typically present.
\end{remark}

   For the harmonic potential \eqref{V_har}, anisotropy in the trapping
frequencies may lead to effective lower-dimensional behavior.
More precisely, when one or two trapping frequencies are much larger than the others, the wave function becomes strongly confined in the corresponding
directions.
As a result, the dynamics in these directions are effectively frozen, and the three-dimensional model can be reduced to a lower-dimensional one. 

We next derive the reduced equations from \eqref{eGPE_nodim}--\eqref{normalization_nodim} under the harmonic potential \eqref{V_har}.


\noindent \textbf{I. Disk-shaped Case.} 
We first consider the disk-shaped case, where the trapping frequencies satisfy
\(\gamma_x \ll \gamma_z\) and \(\gamma_y \ll \gamma_z\).
In this regime, the confinement in the \(z\)-direction is much stronger than
that in the \(x\)- and \(y\)-directions, so that the effective dynamics are
essentially concentrated on the \((x,y)\)-plane.
Formally, this corresponds to the limit \(\gamma_z \to \infty\).
According to \cite{ben2005nonlinear}, under strong confinement along the
\(z\)-axis, the wave function in this direction can be well approximated by a
Gaussian ground mode.
We therefore adopt the factorization
\[
\psi(\mathbf{x},t)=\psi_{12}(x,y,t)\psi_3(z),
\]
where \(\psi_3\) is the normalized Gaussian-type profile in the \(z\)-direction, satisfying
\[
\int_{\mathbb{R}} |\psi_3(z)|^2\,dz = 1.
\]
Hence,
\[
\int_{\mathbb{R}^2} |\psi_{12}(x,y,t)|^2\,dxdy =\int_{\mathbb{R}^3} |\psi(\mathbf{x},t)|^2\,d\mathbf{x}  = c^2.
\]
Substituting this ansatz into \eqref{eGPE_nodim} and integrating over $z$, we obtain the two-dimensional eGP equation
\begin{equation}
    i\partial_t \psi_{12}
    = -\frac{1}{2}\nabla_{\perp}^2\psi_{12}
    + \frac{1}{2}\left(\gamma_x^2 x^2+\gamma_y^2 y^2 + C\right)\psi_{12}
    + \beta_2 |\psi_{12}|^2\psi_{12}
    + \lambda_2 |\psi_{12}|^3\psi_{12},
    \label{eGPE_trans}
\end{equation}
where \(\nabla_{\perp}^2 = \partial_x^2 + \partial_y^2\) denotes the transverse Laplacian,
\[
\beta_2 = \beta\int_{\mathbb{R}} |\psi_3(z)|^4\,dz,
\qquad
\lambda_2 = \lambda\int_{\mathbb{R}} |\psi_3(z)|^5\,dz,
\]
and
\[
C = \gamma_z^2 \int_{\mathbb{R}} z^2 |\psi_3(z)|^2\,dz
    + \int_{\mathbb{R}} \left|\frac{d\psi_3}{dz}\right|^2\,dz
\]
is a constant.
By the phase transformation
\[
\psi_{12} \mapsto \psi_{12} e^{-iCt/2},
\]
the constant \(C\) can be removed from the equation.

\noindent \textbf{II. Cigar-shaped Case.}
We next consider the cigar-shaped case, where the trapping frequencies satisfy
\(\gamma_x \ll \gamma_y\) and \(\gamma_x \ll \gamma_z\).
In this regime, the confinement in the \(y\)- and \(z\)-directions is much
stronger than that in the \(x\)-direction, so that the effective dynamics are
essentially concentrated along the \(x\)-axis.
Formally, this corresponds to the limit \(\gamma_y,\gamma_z \to \infty\).

Under such strong confinement, we assume that the profile in the
\((y,z)\)-plane can be approximated by a time-independent Gaussian ground mode.
We therefore adopt the factorization
\[
\psi(\mathbf{x},t)=\psi_1(x,t)\psi_{23}(y,z),
\]
where \(\psi_{23}\) is the normalized Gaussian-type profile in the \((y,z)\)-plane, satisfying
\[
\int_{\mathbb{R}^2} |\psi_{23}(y,z)|^2\,dydz = 1.
\]
Hence,
\[
 \int_{\mathbb{R}} |\psi_{1}(x)|^2\,dx  = \int_{\mathbb{R}^3} |\psi(\mathbf{x},t)|^2\,d\mathbf{x}  = c^2.
\]
Substituting this ansatz into \eqref{eGPE_nodim} and integrating over $y$ and $z$, we obtain the one-dimensional eGP equation
\begin{equation}
    i\partial_t \psi_1
    = -\frac{1}{2}\frac{\partial^2\psi_1}{\partial x^2}
    + \frac{1}{2}\left(\gamma_x^2 x^2 + C\right)\psi_1
    + \beta_1 |\psi_1|^2\psi_1
    + \lambda_1 |\psi_1|^3\psi_1,
    \label{eGPE_trans_one}
\end{equation}
where
\[
\beta_1 = \beta\int_{\mathbb{R}^2} |\psi_{23}(y,z)|^4\,dydz,
\qquad
\lambda_1 = \lambda\int_{\mathbb{R}^2} |\psi_{23}(y,z)|^5\,dydz,
\]
and
\[
C =  \int_{\mathbb{R}^2} \left( \gamma_y^2 y^2 +  \gamma_z^2 z^2 \right)
|\psi_{23}(y,z)|^2\,dydz
+ \int_{\mathbb{R}^2} \left( |\partial_y\psi_{23}(y,z)|^2
+  |\partial_z\psi_{23}(y,z)|^2 \right)\,dydz
\]
is a constant.
By the phase transformation
\[
\psi_1 \mapsto \psi_1 e^{-iCt/2},
\]
the constant \(C\) can be removed from the equation.
    
For convenience, in the following sections we write the one-, two-, and
three-dimensional eGP equations \eqref{eGPE_nodim}, \eqref{eGPE_trans},
and \eqref{eGPE_trans_one} in the unified form
\begin{equation}
    i\partial_t \psi(\mathbf{x},t)
    = -\frac{1}{2}\nabla^2\psi
    + V(\mathbf{x})\psi
    + \beta |\psi|^2\psi
    + \lambda |\psi|^3\psi, 
    \label{eGPE}
\end{equation}    
with the normalization constraint
\begin{equation}\label{mass}
\|\psi(\cdot,t)\|_2 = c,
\end{equation}
where \(\mathbf{x}\in\mathbb{R}^d\), \(\nabla^2\) denotes the Laplacian in
\(\mathbb{R}^d\), and \(d=1,2,3\).
The corresponding energy functional associated with \eqref{eGPE} is given by
\begin{equation}
\label{energy}
    E(\psi)
    =
    \int_{\mathbb R^d}
    \left(
    \frac{1}{2}|\nabla\psi|^2
    +
    V(\mathbf{x})|\psi|^2
    +
    \frac{\beta}{2}|\psi|^4
    +
    \frac{2\lambda}{5}|\psi|^5
    \right)\,d\mathbf{x}.
\end{equation}
This energy functional will be used to study the existence and nonexistence of ground states under the mass constraint \eqref{mass}.

\section{Existence and nonexistence of ground states}
\label{sec:theory}

In this section, we investigate the existence and nonexistence of ground states
for the unified eGP equation \eqref{eGPE}--\eqref{mass}. A ground state
\(\phi_g\) is defined as a minimizer of the energy functional
\eqref{energy} under the mass constraint \(\|\phi\|_2=c\), namely,
\begin{equation}
	\phi_g \in \arg\min_{\|\phi\|_2=c} E(\phi).
\end{equation}
We begin with the main results.

\subsection{Main results}
To formulate the problem precisely, for \(c>0\), define
\begin{equation} \label{def:S}
S_d(c):=\{u\in H^1(\mathbb R^d):\|u\|_2=c\},
\qquad
\gamma(c):=\inf_{u\in S_d(c)}E(u).
\end{equation}
We then consider two typical situations for the external potential \(V\).
The first is the free-space case
\begin{equation} \label{def:V0}
V(\mathbf{x}) \equiv 0,
\end{equation}
and the second is the case of a confining external potential.
The free-space case is of particular interest for the eGP model, since it describes self-trapped states and has no direct analogue in the conventional trapped BEC setting.

We now state the main existence and nonexistence results for these two cases.

\begin{theorem}[Free-space case]\label{thm:V0}
Let \(d=1,2,3\), \(\lambda>0\) and \(V(\mathbf{x})\equiv0\). Then \(\gamma(c)\) is well defined for every \(c>0\), and the following hold.

\begin{itemize}
\item[(i)] If \(\beta\ge0\), then \(\gamma(c)=0\) for all \(c>0\), and \(E(\cdot)\) admits no minimizer on \(S_d(c)\).

\item[(ii)] If \(\beta<0\), then
\begin{itemize}
\item[(a)] for \(d=1\), \(\gamma(c)<0\) for all \(c>0\), and \(E(\cdot)\) admits a minimizer on \(S_1(c)\) for every \(c>0\);
\item[(b)] for \(d=2,3\), there exists \(c_{d}^*\in(0,\infty)\) such that \(\gamma(c)=0\) for \(c\in(0,c_{d}^*)\), while \(\gamma(c)<0\) for  \(c\in(c_{d}^*,+\infty)\). Moreover, \(E(\cdot)\) admits no minimizer on \(S_d(c)\) for \(c\in(0,c_{d}^*)\), whereas \(E(\cdot)\) admits a minimizer on \(S_d(c)\) for  \(c\in(c_{d}^*,+\infty)\).
\end{itemize}
\end{itemize}
\end{theorem}

\begin{theorem}[Confining potential case]\label{thm:Vhar}
Let \(d=1,2,3\), \(\lambda>0\) and assume that \(V:\mathbb R^d\to[0,\infty)\) satisfies
\begin{equation} \label{def:V_confine}
\lim_{|\mathbf{x}|\to\infty} V(\mathbf{x})=\infty.
\end{equation}
Then, for every \(c>0\) and \(\beta\in\mathbb R\), the functional \(E(\cdot)\) admits a minimizer on \(S_d(c)\).
\end{theorem}

\begin{remark}\label{rem:positive-minimizer}
Whenever a minimizer exists, it may be chosen to be nonnegative.
Indeed, for any \(u\in S_d(c)\), one has \(|u|\in S_d(c)\) and \(E(|u|)\le E(u)\).
In the free-space case \(V(\mathbf{x})\equiv0\), one may further choose a minimizer to be symmetric decreasing
by Schwarz rearrangement.
\end{remark}
    
    \subsection{Free-space case: proof of Theorem \ref{thm:V0}}

In this subsection, we work throughout in the free-space case \(V(\mathbf{x})\equiv 0\).
We first collect some basic properties of the constrained energy functional and establish several auxiliary results concerning \(\gamma(c)\).
Our proof is inspired by the argument in \cite{2021On}. A main difference is that the present model does not contain a dipolar interaction term. This allows us to extend the argument and establish the result in a unified way for \(d=1,2,3\).

\begin{lemma}\label{lem:gamma-basic}
Let \(d=1,2,3\), \(\beta\in\mathbb{R}\), \(\lambda>0\), and \(c>0\). Then the functional \(E(\cdot)\) is bounded from below on \(S_d(c)\). Consequently, \(\gamma(c)\) is well defined and satisfies
\[
\gamma(c)\le 0.
\]
\end{lemma}
    
\begin{proof}
For \(u\in S_d(c)\),
\begin{equation}\label{proof:Lem31_energy}
E(u)=\frac12\|\nabla u\|_2^2+\frac{\beta}{2}\|u\|_4^4+\frac{2\lambda}{5}\|u\|_5^5.
\end{equation}
By interpolation between \(L^2\) and \(L^5\), there exists \(\theta\in(0,1)\) such that
\[
\|u\|_4 \le \|u\|_2^\theta \|u\|_5^{1-\theta} = c^{\theta}\|u\|_5^{1-\theta} .
\]
Hence
\[
\|u\|_4^4 \le c^{4\theta}\|u\|_5^{4(1-\theta)}.
\]
Since \(4(1-\theta)<5\), Young's inequality implies that, for every \(\varepsilon>0\),
\begin{equation}\label{proof:Lem31_young}
\|u\|_4^4 \le \varepsilon \|u\|_5^5 + C_{\varepsilon,c}, 
\end{equation}
where $C_{\varepsilon,c}$ is some constant depending on $c$ and $\varepsilon$. 
Combining \eqref{proof:Lem31_young} and \eqref{proof:Lem31_energy}, we obtain
\begin{equation}\label{proof:Lem31_energy_lowerbound}
E(u)\ge \frac12\|\nabla u\|_2^2
+\Bigl(\frac{2\lambda}{5}-\frac{|\beta|}{2}\varepsilon\Bigr)\|u\|_5^5
-\dfrac{|\beta|}{2}C_{\varepsilon,c}.
\end{equation}
Choosing \(\varepsilon>0\) sufficiently small such that $\frac{2\lambda}{5}-\frac{|\beta|}{2}\varepsilon>0$, we conclude that \(E(\cdot)\) is bounded from below on \(S_d(c)\).

Next, for \(u\in S_d(c)\) and \(s>0\), define the \(L^2\)-preserving scaling 
\begin{equation}\label{eq:L2-scaling}
u^{2,s}(x):=s^{d/2}u(sx).
\end{equation}
Then \(u^{2,s}\in S_d(c)\) and
\begin{equation}  \label{eq:E_L2-scaling}
E(u^{2,s})
=
\frac{s^2}{2}\|\nabla u\|_2^2
+\frac{\beta}{2}s^d\|u\|_4^4
+\frac{2\lambda}{5}s^{3d/2}\|u\|_5^5.
\end{equation}
Hence \(E(u^{2,s})\to0\) as \(s\to0^+\), and therefore \(\gamma(c)\le0\).
\end{proof}

\begin{lemma}\label{lem:gamma-mono}
Assume that \(\beta<0\). Then the map
\[
c\mapsto\gamma(c)
\]
is nonincreasing on \((0,\infty)\). Moreover, it is strictly decreasing on the set
\[
\{c>0:\gamma(c)<0\}.
\]
\end{lemma}
\begin{proof}
For \(u\in H^1(\mathbb R^d)\) and \(s>0\), define the \(L^5\)-preserving scaling
\begin{equation}\label{eq:L5-scaling}
u^{5,s}(x):=s^{-d/5}u(s^{-1}x).
\end{equation}
Then
\begin{equation}
\label{eq:mass_L5-scaling}
\|u^{5,s}\|_2=s^{\frac{3d}{10}}\|u\|_2,
\end{equation}
and
\begin{equation}
\label{eq:E_L5-scaling}
E(u^{5,s})
=
\frac12 s^{\frac{3d}{5}-2} \|\nabla u\|_2^2
+\frac{\beta}{2} s^{\frac{d}{5}} \|u\|_4^4
+\frac{2\lambda}{5}\|u\|_5^5,
\end{equation}
By the expression for the energy functional \eqref{proof:Lem31_energy},
\[
\frac{\beta}{2}\|u\|_4^4
=
E(u)-\frac12\|\nabla u\|_2^2-\frac{2\lambda}{5}\|u\|_5^5,
\]
and substituting it into \eqref{eq:E_L5-scaling}, we obtain
\begin{equation}\label{eq:E_t-L5scaling}
E(u^{5,s})
=
s^{\frac{d}{5}} E(u)
+\frac12(s^{\frac{d}{5}} - s^{\frac{3d}{5}-2})\|\nabla u\|_2^2
-\frac{2\lambda}{5}(s^{\frac{d}{5}}-1)\|u\|_5^5.
\end{equation}
Hence, for every \(s>1\), since \(d=1,2,3\), we have
\[
s^{\frac{d}{5}}-s^{\frac{3d}{5}-2}>0
\quad\text{and}\quad
s^{\frac{d}{5}}-1>0,
\]
and thus
\begin{equation}\label{eq:mono-est}
E(u^{5,s})\le s^{\frac{d}{5}} E(u).
\end{equation}

Now we show that \(c\mapsto \gamma(c)\) is nonincreasing. Let \(0<c_1<c_2\), and let \(\{u_n\}\subset S_d(c_1)\) be a minimizing sequence such that
\begin{equation}
E(u_n)\to\gamma(c_1).
\end{equation}
Set
\begin{equation}\label{proof:Lem_s}
s=\Bigl(\frac{c_2}{c_1}\Bigr)^{\frac{10}{3d}}>1.
\end{equation}
Then 
\[\|u_n^{5,s}\|_2=c_2\quad  \text{and thus} \quad u_n^{5,s}\in S_d(c_2).
\]
 By \eqref{eq:mono-est},
\[
\gamma(c_2)\le E(u_n^{5,s})\le s^{\frac{d}{5}} E(u_n).
\]
Passing to the limit gives
\begin{equation}\label{proof:Lem32_gamma}
\gamma(c_2)\le s^{\frac{d}{5}} \gamma(c_1).
\end{equation}
By \eqref{proof:Lem_s} and Lemma~\ref{lem:gamma-basic}, we have \(s^{\frac{d}{5}}>1\) and \(\gamma(c_1)\le 0\). Hence \eqref{proof:Lem32_gamma} yields
\[
\gamma(c_2)\le \gamma(c_1),
\]
which proves that \(c\mapsto\gamma(c)\) is nonincreasing.

If in addition \(\gamma(c_1)<0\), then by \eqref{proof:Lem_s},
\[
s^{\frac{d}{5}}\gamma(c_1)<\gamma(c_1).
\]
Combining this with \eqref{proof:Lem32_gamma}, we obtain
\[
\gamma(c_2)<\gamma(c_1).
\]
Therefore, \(c\mapsto\gamma(c)\) is strictly decreasing on \(\{c>0:\gamma(c)<0\}\).

\end{proof}

\begin{lemma}\label{lem:small-mass}
Assume that \(\beta<0\). If \(d=2,3\), then there exists \(\hat c_d>0\) such that
\[
E(u)>0 \qquad \text{for all } u\in S_d(c), \quad c\in(0,\hat c_d).
\]
Consequently,
\[
\gamma(c)=0 \qquad \text{for all } c\in(0,\hat c_d),
\]
and \(E(\cdot)\) has no minimizer on \(S_d(c)\) for such \(c\).
\end{lemma}

\begin{proof}
We first consider the case \(d=2\). By the Gagliardo--Nirenberg inequality,
\[
\|u\|_4^4\le C_1\|u\|_2^2\|\nabla u\|_2^2
= C_1c^2\|\nabla u\|_2^2.
\]
Since \(\beta<0\), it follows from \eqref{proof:Lem31_energy} that 
\[
E(u)\ge \frac12(1+\beta C_1c^2)\|\nabla u\|_2^2+\frac{2\lambda}{5}\|u\|_5^5.
\]
The right-hand side is positive for all \(c>0\) sufficiently small.

For \(d=3\), Hölder's inequality and Sobolev's embedding give
\begin{equation}\label{proof:Lem33_holder}
\|u\|_4^4 \le \|u\|_2\|u\|_6^3 \le C_2c\,\|\nabla u\|_2^3,
\end{equation}
while interpolation between \(L^2\) and \(L^5\) yields
\begin{equation}\label{proof:Lem33_inte}
\|u\|_4^4 \le \|u\|_2^{2/3}\|u\|_5^{10/3}
= c^{2/3}\|u\|_5^{10/3}.
\end{equation}
Taking the \(2/5\)-power of \eqref{proof:Lem33_holder} and the \(3/5\)-power of \eqref{proof:Lem33_inte}, and then multiplying the resulting inequalities, we obtain
\begin{equation}\label{proof:Lem33_mix}
\|u\|_4^4
\le
C_2^{2/5}\,c^{4/5}\|\nabla u\|_2^{6/5}\|u\|_5^2.
\end{equation}
Since the right-hand side of \eqref{proof:Lem33_mix} is of the form
\[
c^{4/5}\bigl(\|\nabla u\|_2^2\bigr)^{3/5}\bigl(\|u\|_5^5\bigr)^{2/5},
\]
Young's inequality implies that for every \(\varepsilon>0\),
\begin{equation}\label{proof:Lem33_young}
\|u\|_4^4
\le
\varepsilon \|\nabla u\|_2^2
+
C_{2,\varepsilon}c^2\|u\|_5^5,
\end{equation}
where \(C_{2,\varepsilon}>0\) is a constant depending only on \(C_2\) and \(\varepsilon\).
Combining \eqref{proof:Lem33_young} and \eqref{proof:Lem31_energy}, we obtain
\[
E(u)\ge
\Bigl(\frac12-\frac{|\beta|}{2}\varepsilon\Bigr)\|\nabla u\|_2^2
+
\Bigl(\frac{2\lambda}{5}-\frac{|\beta|}{2}C_{2,\varepsilon} c^2\Bigr)\|u\|_5^5.
\]
Choosing first \(\varepsilon>0\) sufficiently small and then \(c>0\) sufficiently small, we ensure that both coefficients on the right-hand side are positive.  Consequently,
\[
E(u)> 0 \qquad \text{for all } u\in S_3(c),
\]
and hence \(\gamma(c)\ge 0\). Combining this with Lemma~\ref{lem:gamma-basic}, which gives \(\gamma(c)\le 0\), we conclude that \(\gamma(c)=0\) for all sufficiently small \(c>0\).

Thus, for \(d=2,3\), there exists \(\hat c_d>0\) such that, for every \(c\in(0,\hat c_d)\),
\[
E(u)>0 \qquad \text{for all } u\in S_d(c).
\]
By Lemma~\ref{lem:gamma-basic}, \(\gamma(c)\le 0\). Hence
\[
\gamma(c)=0 \qquad \text{for all } c\in(0,\hat c_d).
\]
Since \(E(u)>0\) for every \(u\in S_d(c)\), the infimum \(\gamma(c)\) is not attained on \(S_d(c)\).

\end{proof}

\begin{lemma}\label{lem:gamma-sign}
Assume that \(\beta<0\).

\begin{itemize}
\item[(i)] If \(d=1\), then
\[
\gamma(c)<0 \qquad \text{for all } c>0.
\]

\item[(ii)] If \(d=2,3\), then there exists \(c_{d}^*\in(0,\infty)\) such that
\[
\gamma(c)=0 \quad \text{for } c\in(0,c_{d}^*),\qquad
\gamma(c)<0 \quad \text{for } c\in(c_{d}^*,\infty).
\]
\end{itemize}
\end{lemma}

\begin{proof}
For \(d=1\), let \(u\in S_1(c)\), and let \(u^{2,s}\) be defined by the \(L^2\)-preserving scaling \eqref{eq:L2-scaling}. Then \(u^{2,s}\in S_1(c)\), and \eqref{eq:E_L2-scaling} becomes
\begin{equation}\label{proof:Lem34_energy1}
E(u^{2,s})
=
\frac{s^2}{2}\|u'\|_2^2
+\frac{\beta}{2}s\|u\|_4^4
+\frac{2\lambda}{5}s^{3/2}\|u\|_5^5.
\end{equation}
Since \(u\in S_1(c)\) with \(c>0\), we have \(u\not\equiv0\), and hence \(\|u\|_4>0\). It follows from \(\beta<0\) and \eqref{proof:Lem34_energy1} that
\[
E(u^{2,s})=\frac{\beta}{2}s\|u\|_4^4+o(s)<0
\qquad\text{as } s\to0^+.
\]
Thus \(\gamma(c)\le E(u^{2,s})<0\), proving \((i)\).

For \(d=2,3\), Lemma~\ref{lem:small-mass} yields \(\hat c_d>0\) such that
\[
\gamma(c)=0 \qquad \text{for all } c\in(0,\hat c_d).
\]
On the other hand, fix a nontrivial function \(u\in C_c^\infty(\mathbb R^d)\), and let \(u^{5,s}\) be defined by the \(L^5\)-preserving scaling \eqref{eq:L5-scaling}. 
Then
\begin{equation}\label{proof:Lem34_energy23}
E(u^{5,s})
=
\frac12 s^{\frac{3d}{5}-2}\|\nabla u\|_2^2
+\frac{\beta}{2}s^{\frac{d}{5}}\|u\|_4^4
+\frac{2\lambda}{5}\|u\|_5^5.
\end{equation}
Since \(u\in S_d(c)\) with \(c>0\), we have \(u\not\equiv0\), and hence \(\|u\|_4>0\). Since \(\frac{3d}{5}-2<0\) for \(d=2,3\), the gradient term in \eqref{proof:Lem34_energy23} tends to \(0\) as \(s\to\infty\), while the quartic term tends to \(-\infty\) when \(\beta<0\). Therefore,
\[
E(u^{5,s})\to -\infty \qquad \text{as } s\to+\infty.
\]
In particular, there exists \(s_0>0\) such that
\[
E(u^{5,s_0})<0.
\]
Let
\[
c_0:=\|u^{5,s_0}\|_2=s_0^{\frac{3d}{10}}\|u\|_2>0.
\]
Then
\[
\gamma(c_0)\le E(u^{5,s_0})<0.
\]
Define
\begin{equation}\label{proof:Lem34_cds}
c_{d}^*:=\inf\{c>0:\gamma(c)<0\}.
\end{equation}
Since \(\gamma(\hat c_d)=0\), \(\gamma(c_0)<0\), and \(c\mapsto\gamma(c)\) is nonincreasing by Lemma~\ref{lem:gamma-mono}, we obtain
\[
c_{d}^*\in[\hat c_d,c_0]\subset(0,\infty).
\]
By the same monotonicity property, we further deduce that
\[
\gamma(c)=0 \quad \text{for } c\in(0,c_{d}^*),\qquad
\gamma(c)<0 \quad \text{for } c\in(c_{d}^*,\infty).
\]
This proves \((ii)\).
\end{proof}

We are now in a position to prove Theorem \ref{thm:V0}. The above lemmas determine the basic properties, monotonicity, and sign structure of \(\gamma(c)\), from which the existence and nonexistence assertions follow.

\begin{proof}[Proof of Theorem \ref{thm:V0}]
\textbf{Case 1.} \(\beta\ge0\).
For every \(u\in S_d(c)\),
\[
E(u)=\frac12\|\nabla u\|_2^2+\frac{\beta}{2}\|u\|_4^4+\frac{2\lambda}{5}\|u\|_5^5 \,> 0.
\]
On the other hand, Lemma \ref{lem:gamma-basic} gives \(\gamma(c)\le0\). Hence
\[
\gamma(c)=0 \qquad \text{for all } c>0.
\]
Since \(E(u)>0\) for all \(u\in S_d(c)\), the infimum \(0\) is not attained on \(S_d(c)\). Therefore, \(E(\cdot)\) admits no minimizer on \(S_d(c)\).

\textbf{Case 2.} \(\beta<0\) and \(\gamma(c)=0\). 
By Lemma~\ref{lem:gamma-sign}, this occurs when \(d=2,3\) and \(c\in(0,c_{d}^*)\). We show that \(E\) has no minimizer on \(S_d(c)\) for \(c\in(0,c_{d}^*)\). Indeed, if \(u\in S_d(c)\) were a minimizer, then \(E(u)=0\). For any \(s>1\), let \(u^{5,s}\) be defined by \eqref{eq:L5-scaling}. 
Recalling \eqref{eq:mass_L5-scaling} and \eqref{eq:mono-est}, we have
\begin{equation}\label{case2:energy}
\|u^{5,s}\|_2=s^{\frac{3d}{10}}c, 
\qquad E(u^{5,s})< s^{\frac{d}{5}} E(u)=0.
\end{equation}
Since \(c<c_d^*\) and \(\|u^{5,s}\|_2\to c\) as \(s\to1^+\), there exists \(s>1\) such that
\begin{equation}\label{case2:cds}
c<\|u^{5,s}\|_2<c_d^*.
\end{equation}
By \eqref{case2:cds} and Lemma~\ref{lem:gamma-sign}, it follows that
\begin{equation}\label{case2:gamma0}
\gamma(\|u^{5,s}\|_2)=0.
\end{equation}
On the other hand, since \(u^{5,s}\in S_d(\|u^{5,s}\|_2)\), by the definition of \(\gamma\) in \eqref{def:S} and \eqref{case2:energy},
\begin{equation}\label{case2:gamma}
\gamma(\|u^{5,s}\|_2)\le E(u^{5,s})<0.
\end{equation}
This contradicts \eqref{case2:gamma0}. Hence no minimizer exists on \(S_d(c)\) for \(c\in(0,c_d^*)\).

\textbf{Case 3.} \(\beta<0\) and \(\gamma(c)<0\).
By Lemma~\ref{lem:gamma-sign}, this occurs for all \(c>0\) when \(d=1\), and for all \(c>c_d^*\) when \(d=2,3\).
Denote by
\[
\mathcal R_d
:=
\bigl\{
u^* : u^* \text{ is the symmetric decreasing rearrangement of } |u|,
\ u\in H^1(\mathbb R^d)
\bigr\}.
\]
Let \(\{u_n\}\subset S_d(c)\) be a minimizing sequence, and for each \(n\) let
\(u_n^*\in\mathcal R_d\) denote the symmetric decreasing rearrangement of \(|u_n|\). Then \(u_n^*\ge0\) and
\begin{equation}\label{case3:norm245}
\|u_n^*\|_p=\|u_n\|_p, \qquad p=2,4,5.
\end{equation}
Moreover, by the Pólya--Szegő inequality,
\begin{equation}\label{case3:nabla_norm}
\|\nabla u_n^*\|_2\le \|\nabla u_n\|_2.
\end{equation}
Combining \eqref{case3:norm245}, \eqref{case3:nabla_norm}, and \eqref{energy},
it follows that
\[
E(u_n^*)\le E(u_n).
\]
Hence \(\{u_n^*\}\subset S_d(c)\) is still a minimizing sequence. Replacing \(u_n\) by \(u_n^*\) if necessary, we may therefore assume that
\[
u_n\in \mathcal R_d \qquad \text{for all } n.
\]
For this sequence, we have
\[
E(u_n)\to \gamma(c)
\qquad \text{as } n\to\infty.
\]
Since \(\gamma(c)>-\infty\), the sequence \(\{E(u_n)\}\) is bounded. Recalling \eqref{proof:Lem31_energy_lowerbound} from the proof of Lemma~\ref{lem:gamma-basic}, we infer that \(\{\nabla u_n\}\) is bounded in \(L^2(\mathbb R^d)\). Therefore, \(\{u_n\}\) is bounded in \(H^1(\mathbb R^d)\).
Hence, up to a subsequence,
\[
u_n \rightharpoonup u
\qquad \text{weakly in } H^1(\mathbb R^d).
\]
Since the embedding \(\mathcal R_d \hookrightarrow L^4(\mathbb R^d)\cap L^5(\mathbb R^d)\) is compact, we further have
\[
u_n \to u
\qquad \text{strongly in } L^4(\mathbb R^d)\cap L^5(\mathbb R^d).
\]
Therefore, by weak lower semicontinuity of the gradient term and the strong convergence of the nonlinear terms,
\begin{equation}\label{ineq:E_u}
E(u)\le \liminf_{n\to\infty} E(u_n)=\gamma(c).
\end{equation}

Next we show that \(u\in S_d(c)\). Let \(a:=\|u\|_2\). Since the above compact embedding does not include \(L^2(\mathbb R^d)\), we only know from the weak convergence \(u_n\rightharpoonup u\) in \(L^2(\mathbb R^d)\) and \(\|u_n\|_2=c\) that \(a\le c\). Hence
\[
\gamma(a)\le E(u)\le \gamma(c)\le \gamma(a),
\]
where the first inequality follows from the definition of \(\gamma(\cdot)\) \eqref{def:S}, the second from \eqref{ineq:E_u}, and the last from Lemma~\ref{lem:gamma-mono}. Thus
\[
\gamma(a)=\gamma(c).
\]
If \(a<c\), then \(\gamma(a)=\gamma(c)<0\), which contradicts the strict monotonicity of \(\gamma\) on \(\{c>0:\gamma(c)<0\}\) established in Lemma~\ref{lem:gamma-mono}. 
Therefore \(a=c\), so \(u\in S_d(c)\), and consequently
\[
E(u)=\gamma(c).
\]
So \(u\) is a minimizer of \(E(\cdot)\) on \(S_d(c)\).
This completes the proof of the theorem.
\end{proof}

\subsection{Confining potential case: proof of Theorem \ref{thm:Vhar}}

For the case with an external potential, we work in the space
\[
L_V(\mathbb R^d):=\left\{u\in L^2(\mathbb R^d): \int_{\mathbb R^d}V(x)|u(x)|^2\,dx<\infty\right\},
\qquad
X:=H^1(\mathbb R^d)\cap L_V(\mathbb R^d).
\]
For later use, we recall the following compactness result from \cite{2013Mathematical}.

\begin{lemma}\label{lem:compact-X}
Assume that \(V:\mathbb R^d\to[0,\infty)\) satisfies
\[
\lim_{|x|\to\infty}V(x)=\infty.
\]
Then the embedding \(X\hookrightarrow L^q(\mathbb R^d)\) is compact for
\[
q\in[2,\infty)\quad \text{if } d=1,2,
\qquad
q\in[2,6)\quad \text{if } d=3.
\]
\end{lemma}

\begin{proof}[Proof of Theorem \ref{thm:Vhar}]

By a similar argument to that in Lemma~\ref{lem:gamma-basic}, we have \(\gamma(c)>-\infty\). Let \(\{u_n\}\subset S_d(c)\) be a minimizing sequence such that
\[
E(u_n)\to \gamma(c).
\]
Since \(V\ge0\), the same argument as in the proof of the free-space case shows that \(\{u_n\}\) is bounded in \(H^1(\mathbb R^d)\) and in \(L_V(\mathbb R^d)\). Therefore, \(\{u_n\}\) is bounded in \(X\). Hence, up to a subsequence,
\[
u_n\rightharpoonup u_0 \qquad \text{weakly in } X.
\]
By Lemma~\ref{lem:compact-X}, we have
\[
u_n\to u_0 \qquad \text{strongly in } L^2(\mathbb R^d)\cap L^4(\mathbb R^d)\cap L^5(\mathbb R^d).
\]
In particular, \(u_0\in S_d(c)\), and therefore
\[
\gamma(c)\le E(u_0).
\]
Moreover, by the weak lower semicontinuity of the \(H^1\)- and \(L_V\)-terms and the strong convergence of the \(L^4\)- and \(L^5\)-terms, we have
\[
E(u_0)\le \liminf_{n\to\infty}E(u_n)=\gamma(c).
\]
Therefore,
\[
E(u_0)=\gamma(c),
\]
and thus \(u_0\) is a minimizer of \(E(\cdot)\) on \(S_d(c)\). Moreover, the minimizer may be chosen nonnegative. Indeed, since \(|u_0|\in S_d(c)\) and \(E(|u_0|)\le E(u_0)=\gamma(c)\), it follows that
\[
E(|u_0|)=\gamma(c).
\]
Hence \(|u_0|\) is also a minimizer.
\end{proof}

\section{Numerical methods for computing ground states} \label{sec:scheme}

In this section, we present numerical methods for computing ground states of the eGP model \eqref{eGPE}--\eqref{mass}. We adopt a gradient-flow iteration with normalization, a standard and effective approach for nonlinear Schr\"odinger-type equations. We first introduce the time semi-discrete formulation, and then combine it with a finite element discretization in space to obtain a fully discrete method. The resulting framework applies in a unified manner to one-, two-, and three-dimensional problems.

\subsection{Gradient-flow iteration}

Since the model contains both cubic and higher-order nonlinear terms, we aim to construct an iteration that remains simple and stable in different parameter regimes. To this end, we introduce the following auxiliary gradient flow in imaginary time with a Lagrange multiplier \cite{2013Mathematical, bao2004computing}
\begin{equation}
    \partial_t \phi(\mathbf{x},t)
    =
    \frac{1}{2}\nabla^2\phi
    -
    V(\mathbf{x})\phi
    -
    \beta |\phi|^2\phi
    -
    \lambda |\phi|^3\phi
    +
    \mu_\phi(t)\phi,
    \quad
    \mathbf{x}\in\mathbb R^d,
    \,
    t>0,
    \label{eq:gfdn_continuous}
\end{equation}
which can be viewed as the constrained \(L^2\)-gradient flow of the energy functional \eqref{energy}.  Here  \(\mu_\phi(t)\) is the Lagrange multiplier associated with the mass constraint \eqref{mass}. The long-time limit of
 \(\phi(\mathbf{x},t)\)  yields the desired ground state.
It is proved that the flow \eqref{eq:gfdn_continuous} preserves the mass and dissipates the energy \cite{2013Mathematical}. 

The flow \eqref{eq:gfdn_continuous} serves as the continuous prototype of our iterative method. To obtain a practical semi-discrete numerical scheme, we discretize the time variable and enforce the mass constraint by a normalization step at each iteration.

Let \(t_n=n\tau\), where \(n\in \mathbb{N}\) and \(\tau>0\) is the time step size, and let
\(\phi^n(\mathbf{x})\) be an approximation to \(\phi(\mathbf{x},t_n)\).
At each iteration, we discretize the Lagrange multiplier in a way consistent with the Euler--Lagrange equation under the mass constraint \(\|\phi\|_2=c\) \cite{liu2021normalized}
\begin{equation}
    \mu^n
    =
    \frac{1}{c^2}\int_{\mathbb R^d}
    \Bigl(
    \frac{1}{2}|\nabla \phi^n|^2
    +
    V(\mathbf{x})|\phi^n|^2
    +
    \beta |\phi^n|^4
    +
    \lambda |\phi^n|^5
    \Bigr)\,d\mathbf{x},
    \label{eq:mu_n}
\end{equation}
and then compute the intermediate state \(\tilde{\phi}^{n+1}\) from
\begin{equation}
    \frac{\tilde{\phi}^{n+1}-\phi^n}{\tau}
    =
    \Bigl(
    \frac{1}{2}\nabla^2
    -
    V(\mathbf{x})
    -
    \lambda |\phi^n|^3
    \Bigr)\tilde{\phi}^{n+1}
    +
    \mu^n \phi^\ast
    -
    \beta |\phi^n|^2 \phi^{**},
    \label{eq:semi_discrete_scheme}
\end{equation}
where
\begin{equation}
    \phi^\ast=
    \begin{cases}
    \tilde{\phi}^{n+1}, & \mu^n<0,\\[4pt]
    \phi^n, & \mu^n\ge0,
    \end{cases}
    \qquad
    \phi^{**}=
    \begin{cases}
    \phi^n, & \beta<0,\\[4pt]
    \tilde{\phi}^{n+1}, & \beta\ge 0.
    \end{cases}
    \label{eq:phi_star}
\end{equation}
Finally, we normalize the intermediate state $\tilde{\phi}^{n+1}$ to obtain
\begin{equation}
    \phi^{n+1}(\mathbf{x})
    =
    c\,\frac{\tilde{\phi}^{n+1}(\mathbf{x})}
    {\|\tilde{\phi}^{n+1}\|_2}.
    \label{eq:semi_discrete_normalization}
\end{equation}
With the above implicit-explicit treatment, each iteration reduces to solving a linear elliptic problem for \(\tilde{\phi}^{n+1}\).

Moreover, the above semi-discrete scheme \eqref{eq:mu_n}--\eqref{eq:semi_discrete_normalization} preserves nonnegativity.

\begin{proposition}[Nonnegativity preservation]
\label{thm:positivity}
Assume that \(\lambda\ge0\) and \(V(\mathbf{x})\ge0\). If
\(\phi^n(\mathbf{x})\ge0\), then the intermediate state
\(\tilde{\phi}^{n+1}\) determined by
\eqref{eq:mu_n}--\eqref{eq:phi_star} satisfies
\[
\tilde{\phi}^{n+1}(\mathbf{x})\ge0.
\]
Consequently, the normalized iterate \(\phi^{n+1}\) defined by
\eqref{eq:semi_discrete_normalization} also satisfies
\[
\phi^{n+1}(\mathbf{x})\ge0.
\]

\end{proposition}

\begin{proof}
Define
\[
(\mu^n)^-:=\min\{\mu^n,0\},\qquad (\mu^n)^+:=\max\{\mu^n,0\},
\]
and
\[
\beta^-:=\min\{\beta,0\},\qquad \beta^+:=\max\{\beta,0\}.
\]
Then, by \eqref{eq:semi_discrete_scheme} and \eqref{eq:phi_star}, \(\tilde{\phi}^{n+1}\) satisfies
\begin{equation}\label{eq:unified-elliptic}
\Bigl(
\frac1\tau-\frac12\Delta+V(\mathbf x)+\lambda|\phi^n|^3+\beta^+|\phi^n|^2-(\mu^n)^-
\Bigr)\tilde{\phi}^{n+1}
=
\Bigl(
\frac1\tau+(\mu^n)^+ - \beta^-|\phi^n|^2
\Bigr)\phi^n.
\end{equation}
Hence each iteration reduces to a linear elliptic problem. 
Moreover, the coefficient multiplying \(\tilde{\phi}^{n+1}\) on the left-hand side is nonnegative, and the right-hand side is nonnegative since \(\phi^n\ge0\). Therefore, by the maximum principle, we obtain \(\tilde{\phi}^{n+1}\ge0\), and the normalization step further implies \(\phi^{n+1}\ge0\).

\end{proof}

\subsection{Finite element discretization}

To obtain a fully discrete scheme, we combine the semi-discrete iteration
\eqref{eq:mu_n}--\eqref{eq:semi_discrete_normalization}
with a finite element discretization in space. Such a choice is well suited to the present problem, since droplet-like or other strongly localized states may contain narrow transition layers, making local spatial resolution important. It also provides a convenient framework for adaptive mesh refinement in higher-dimensional computations. 

Let \(\Omega\subset\mathbb R^d\) be a sufficiently large bounded computational domain, and impose homogeneous Dirichlet boundary conditions on \(\partial\Omega\). Let \(\mathcal T_h\) be a conforming triangulation of \(\Omega\), and define the finite element space
\[
V_h
:=
\left\{
v_h\in C(\bar\Omega):
\ v_h|_K\in \mathbb P_1(K)\ \text{for all }K\in\mathcal T_h,
\ \text{and }v_h=0\ \text{on }\partial\Omega
\right\}
\subset H_0^1(\Omega).
\]

Given \(\phi_h^n\in V_h\) with \(\|\phi_h^n\|_2=c\), we first compute the discrete Lagrange multiplier
\begin{equation}
    \mu_{h}^n
    =
    \frac{1}{c^2}\int_{\Omega}
    \Bigl(
    \frac{1}{2}|\nabla \phi_h^n|^2
    +
    V(\mathbf{x})|\phi_h^n|^2
    +
    \beta|\phi_h^n|^4
    +
    \lambda|\phi_h^n|^5
    \Bigr)\,d\mathbf{x},
    \label{eq:mu_n_fem}
\end{equation}
and then seek the intermediate state \(\tilde{\phi}_h^{n+1}\in V_h\) such that
\begin{equation}
\begin{aligned}
    \Bigl(
    \frac{\tilde{\phi}_h^{n+1}-\phi_h^n}{\tau},\,v_h
    \Bigr)
    &=
    -\frac{1}{2}\bigl(\nabla \tilde{\phi}_h^{n+1},\nabla v_h\bigr)
    -\bigl(V(\mathbf{x})\tilde{\phi}_h^{n+1},v_h\bigr)
    -\bigl(\lambda |\phi_h^n|^3\tilde{\phi}_h^{n+1},v_h\bigr) \\
    &\quad
    +\mu_h^n(\phi_h^\ast,v_h)
    -\bigl(\beta |\phi_h^n|^2\phi_h^{**},v_h\bigr),
    \qquad
    \forall\,v_h\in V_h,
\end{aligned}
\label{eq:fem_scheme}
\end{equation}
where
\begin{equation}
    \phi_h^\ast=
    \begin{cases}
    \tilde{\phi}_h^{n+1}, & \mu_h^n<0,\\[4pt]
    \phi_h^n, & \mu_h^n\ge0,
    \end{cases}
    \qquad
    \phi_h^{**}=
    \begin{cases}
    \phi_h^n, & \beta<0,\\[4pt]
    \tilde{\phi}_h^{n+1}, & \beta\ge 0.
    \end{cases}
    \label{eq:phi_star_fem}
\end{equation}
Finally, we normalize the intermediate state by
\begin{equation}
    \phi_h^{n+1}
    =
    c\,\frac{\tilde{\phi}_h^{n+1}}
    {\|\tilde{\phi}_h^{n+1}\|_2}.
    \label{eq:fem_normalization}
\end{equation}

At each iteration, the above scheme requires solving a linear variational problem for \(\tilde{\phi}_h^{n+1}\). The formulation applies in a unified way to one-, two-, and three-dimensional problems. In practical computations, the problem on \(\mathbb R^d\) is truncated to a sufficiently large bounded domain \(\Omega\), so that the truncation effect is negligible. For higher-dimensional computations, adaptive meshes are employed to better resolve localized structures while keeping the computational cost under control. In this work, the finite element computations are implemented via FreeFEM \cite{hecht2005freefem++}.

\begin{remark}
The above numerical framework also covers the standard Gross--Pitaevskii equation as a special case by setting \(\lambda=0\).
\end{remark}

\begin{remark}
For problems that are one-dimensional, or can be reduced to one-dimensional form under symmetry assumptions, standard finite difference discretizations are already sufficient and often more convenient in practice. The finite element formulation presented above is adopted mainly as a unified framework for one-, two-, and three-dimensional computations, especially for higher-dimensional adaptive computations. For symmetric problems, the corresponding finite difference discretization is given in \ref{app:sym}.
\end{remark}

 
\section{Numerical simulations}
\label{sec:experiment}

In this section, we present numerical results for the ground states of the eGP model \eqref{eGPE}--\eqref{mass}. We first study how the ground-state profiles depend on the model parameters in symmetric settings. We then classify the different regimes in the parameter plane and use a simple heuristic approximation to explain the flat-top behavior observed in the droplet regime. 
Finally, we present several two- and three-dimensional computations to illustrate more general spatial structures.

\subsection{Parameter dependence in symmetric settings}

In the symmetric cases considered here, the two- and three-dimensional ground-state problems can be reduced by symmetry to a one-dimensional formulation. This significantly lowers the computational cost and makes it possible to use fine spatial grids, which is convenient for a systematic study of parameter effects. The corresponding one-dimensional scheme is described in \ref{app:sym}.

 In this subsection, we examine the influence of the mass \(c\), the cubic interaction coefficient \(\beta\), and the LHY coefficient \(\lambda\) in symmetric settings. We first study how the ground-state profile changes with \(c\) for fixed \(\beta\) and \(\lambda\). We then fix the normalized mass \(c=1\) and investigate the effects of varying \(\beta\) and \(\lambda\) in both the free-space case and the harmonic-potential case.

\begin{example}[Effect of the mass \(c\)]
In this example, we study the effect of the mass \(c\) on three-dimensional spherically symmetric ground states with fixed parameters \(\beta=-10\) and \(\lambda=0.1\).

We consider two external potentials, namely the free-space case \(V(r)\equiv 0\) and the harmonic-potential case \(V(r)=10^5 r^2\). The reduced radial problem is solved on \([0,1]\) by the finite difference scheme \eqref{eq:radial_scheme} with \(M=2048\) and \(\tau=10^{-2}\), starting from a normalized Gaussian initial guess. The iteration is terminated when
\(
\max_{0\le j\le M-1}\bigl|\phi_{j+\frac12}^{\,n+1}-\phi_{j+\frac12}^{\,n}\bigr|<10^{-10}.
\)
Figure~\ref{fig:c} shows the computed ground-state profiles for \(c=1,4,8,12\). 

In both the free-space and harmonic-potential cases, increasing \(c\) mainly broadens the central region, while the peak value changes only slightly. In the free-space case, the profile gradually develops a wider flat-top structure, indicating a clear saturation effect characteristic of droplet-like states. In the harmonic-potential case, however, the external confinement makes the profile slightly more localized, and the central plateau is less pronounced than in free space.

\end{example}

\begin{figure}[htp!]
    \centering
    \includegraphics[width=0.99\linewidth]{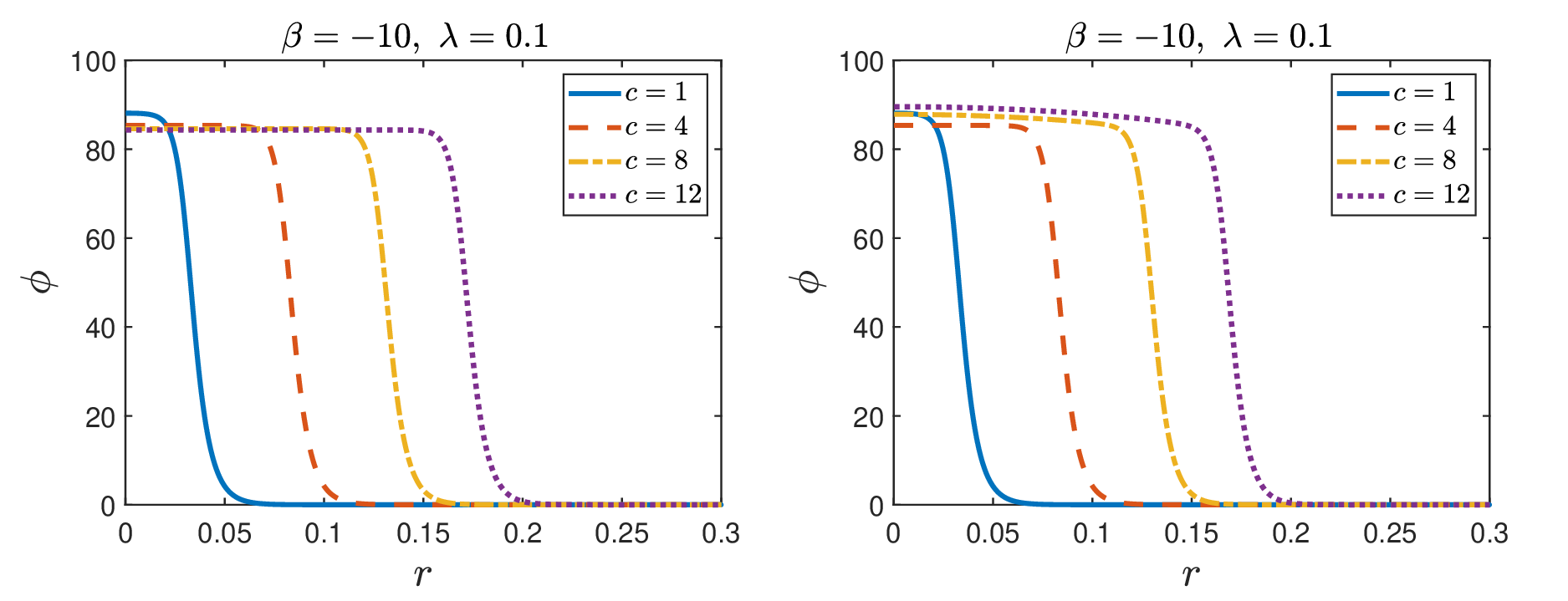}
  \caption{Radial 3D ground-state profiles for different masses \(c=1,4,8,12\) with \(\beta=-10\) and \(\lambda=0.1\). The left panel shows the free-space case, and the right panel shows the harmonic-potential case \(V(r)=10^5r^2\).}
     \label{fig:c}
\end{figure}

In all subsequent numerical experiments, we fix the mass at \(c=1\), so that the effects of \(\beta\) and \(\lambda\) can be compared more directly.

\begin{remark}
A problem with general mass \(c\) can be reduced to the normalized case \(c=1\) by a simple scaling. Under this scaling, the parameters become
\[
\tilde{\beta}=\beta c^2,\qquad
\tilde{\lambda}=\lambda c^3.
\]
Therefore, fixing \(c=1\) does not reduce the generality of the numerical experiments.
\end{remark}

\begin{example}[Effects of \(\beta\) and \(\lambda\)]\label{example:52}
In this example, we study the effects of \(\beta\) and \(\lambda\) on three-dimensional spherically symmetric ground states with normalized mass \(c=1\).

We consider the same two external potentials as in the previous example, namely the free-space case \(V(r)\equiv 0\) and the harmonic-potential case \(V(r)=10^5 r^2\). The reduced radial problem is solved with the same numerical setting as in the previous example, including the initial guess and the stopping criterion. We first fix \(\lambda=4\) and vary \(\beta\), and then fix \(\beta=-10\) and vary \(\lambda\). The resulting profiles are displayed in Figure~\ref{fig:lambda} and Figure~\ref{fig:beta}, respectively.

Figure~\ref{fig:lambda} shows that, for fixed \(\lambda=4\), the ground-state profile becomes increasingly concentrated as \(\beta\) decreases from \(-20\) to \(-50\). More precisely, the central value increases, while the support of the profile shrinks. Figure~\ref{fig:beta} shows that, for fixed \(\beta=-10\), decreasing \(\lambda\) produces a similar effect. In this case, the profile also becomes higher near the center and more localized in space. These observations are consistent with the balance between the attractive cubic term and the repulsive higher-order term. A stronger mean-field attraction or a weaker higher-order repulsion both enhance concentration of the ground state. Moreover, in the harmonic-potential case, the external confinement further compresses the profile, so that the peak is higher and the support is slightly smaller than in the free-space case.

\end{example}

\begin{figure}[htp!]
    \centering
    \includegraphics[width=0.99\linewidth]{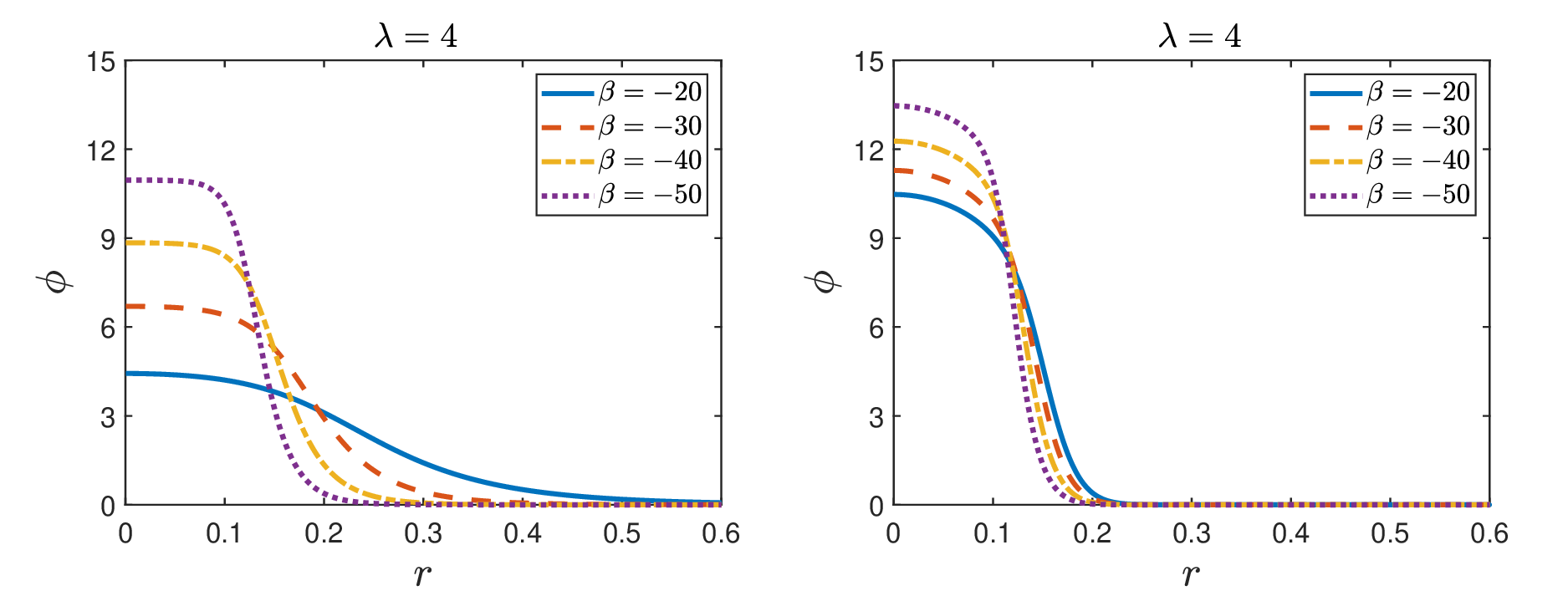}
     \caption{Radial 3D ground-state profiles for \(\beta=-20,-30,-40,-50\) with \(\lambda=4\). The left panel shows the free-space case, and the right panel shows the harmonic-potential case \(V(r)=10^5r^2\). }
         \label{fig:lambda}
\end{figure}

\begin{figure}[htp!]
    \centering
    \includegraphics[width=0.99\linewidth]{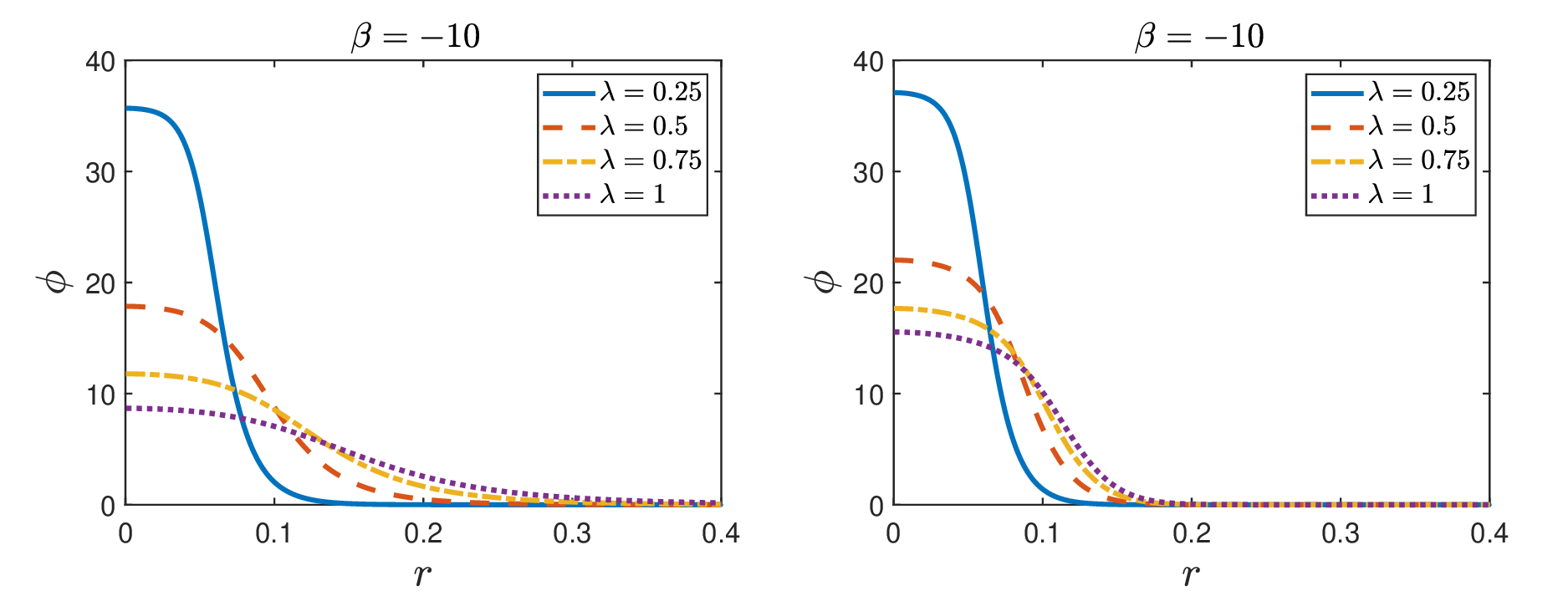}
   \caption{Radial 3D ground-state profiles for \(\lambda=0.25,0.5,0.75,1\) with fixed \(\beta=-10\). The left panel shows the free-space case, and the right panel shows the harmonic-potential case \(V(r)=10^5r^2\). }
       \label{fig:beta}
\end{figure}

\subsection{Phase diagram and heuristic flat-top approximation}

The symmetric computations above show two qualitatively different types of profiles in the parameter plane. Some profiles decay smoothly from the center, while others contain a more pronounced central high-density region. To distinguish these behaviors quantitatively, we introduce the following numerical indicator based on the density
\[
\rho(\mathbf{x}) = |\phi(\mathbf{x})|^2,
\qquad \mathbf{x}\in \mathbb{R}^d.
\]
For a parameter \(\theta\in(0,1)\) chosen close to \(1\), define
\[
\Omega_\theta
=
\left\{
\mathbf{x}\in\mathbb{R}^d:\,
\rho(\mathbf{x})\ge \theta \|\rho\|_{L^\infty}\right\}.
\]
We then set
\[
\eta_\theta
=
\frac{\int_{\Omega_\theta} \rho(\mathbf{x})\,d\mathbf{x}}
{\int_{\mathbb{R}^d} \rho(\mathbf{x})\,d\mathbf{x}},
\]
which measures the fraction of the total mass contained in the region where the density remains close to its maximum.

Profiles with a more evident central high-density core typically yield larger values of \(\eta_\theta\), whereas smoothly decaying profiles give smaller values. Thus, \(\eta_\theta\) serves as a convenient numerical diagnostic for distinguishing different regimes in the parameter plane. 

\begin{example}[Phase diagram in the parameter plane]
In this example, we consider the three-dimensional free-space case \(V(\mathbf x)\equiv 0\) with normalized mass \(c=1\), and use the indicator \(\eta_\theta\) introduced above to classify the profile regimes in the \((\beta,\lambda)\)-plane. For each pair \((\beta,\lambda)\), we compute the corresponding ground state with the same numerical setting as in Example \ref{example:52}, evaluate \(\eta_\theta\) with \(\theta=0.99\), and plot the resulting heatmap.

Figure~\ref{fig:heatmap} shows the distribution of \(\eta_\theta\) over the parameter range
\[
-25<\beta<-1,
\qquad
1/500<\lambda<1.
\]
The phase diagram reveals three qualitatively different regimes. The lower gray region corresponds to parameter values for which no ground state exists in the free-space case. This is consistent with Theorem~\ref{thm:V0}, which predicts a threshold between the existence and nonexistence regimes. Above this region lies a soliton-like regime, where the ground states remain localized but do not exhibit a clear flat-top structure. For stronger attraction or weaker higher-order repulsion, the profiles enter a droplet-like regime, characterized by a more evident high-density core and flat-top behavior. The two marked parameter sets \(A\) and \(B\) illustrate representative droplet-like and soliton-like profiles, respectively.
\end{example}

\begin{figure}[htp!]
    \centering
    \includegraphics[width=0.99\linewidth]{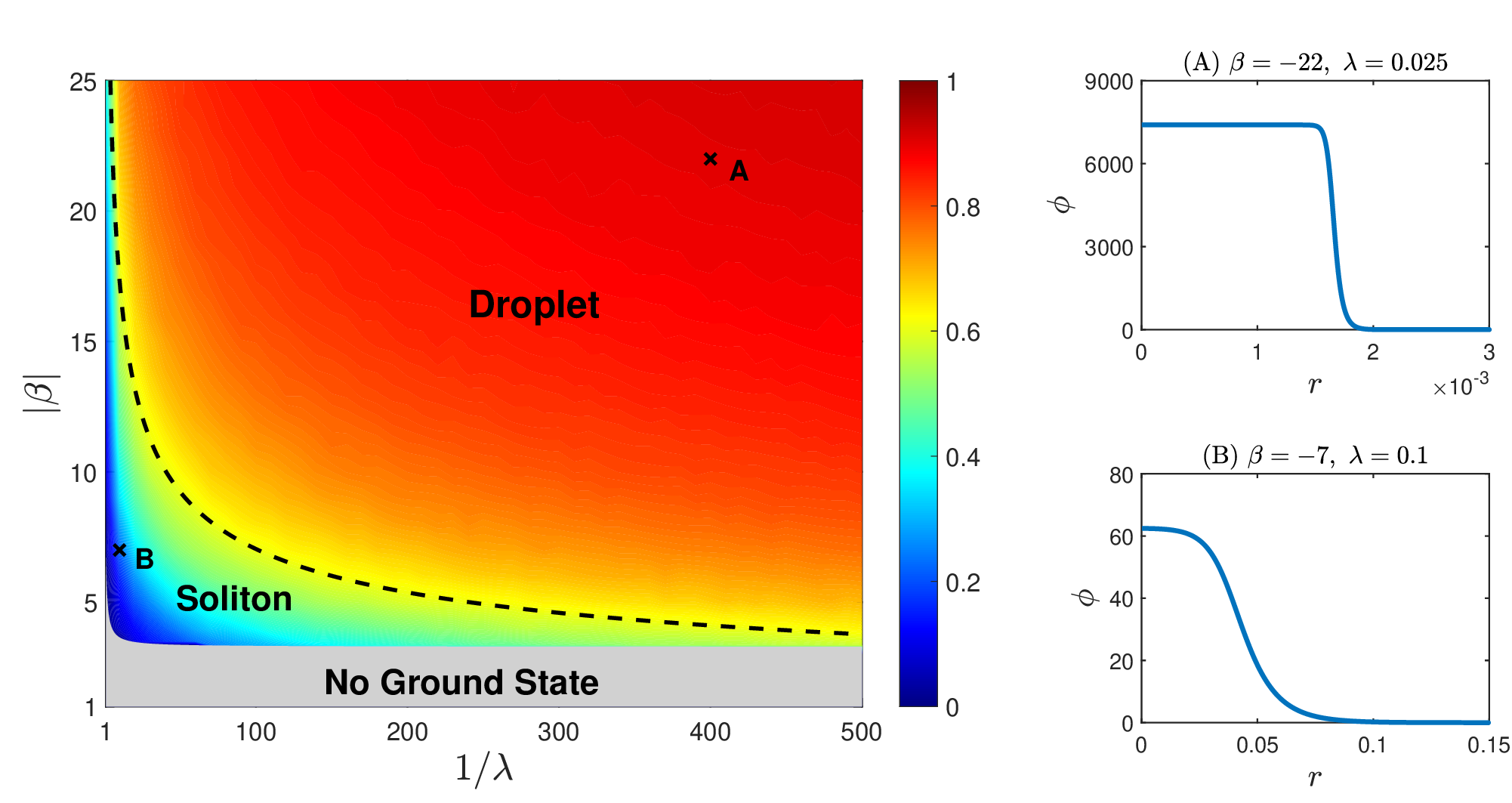}
    \caption{Heatmap of \(\eta_\theta\) in the \((1/\lambda,|\beta|)\)-plane for \(-25<\beta<-1\) and \(1/500<\lambda<1\), showing the droplet-like, soliton-like, and no-ground-state regimes. The dashed curve denotes the contour line corresponding to the threshold \(\eta_\theta=0.62\) used to separate the droplet-like and soliton-like regions. Two representative radial ground-state profiles corresponding to the marked points \(A\) and \(B\) are displayed on the right.}
    \label{fig:heatmap}
\end{figure}

The phase diagram suggests that, in part of the parameter plane, the ground states develop a pronounced central high-density region. This motivates the following heuristic approximation.

In the limiting regime \(\beta\to -\infty\) or \(\lambda\to 0\), the ground state is expected to be well approximated by a compactly supported profile that is nearly constant on its support. Accordingly, we consider the ansatz
\begin{equation}\label{example:54_phi}
\phi(\mathbf{x})\approx a\,\chi_D(\mathbf{x}),
\end{equation}
where \(a>0\) denotes the constant value on the support region \(D\subset\mathbb R^d\), and \(|D|\) denotes the volume of \(D\).
Since \(\|\phi\|_2=c\), the mass constraint gives
\begin{equation}\label{example:54_cons}
a|D|^{1/2}=c.
\end{equation}
Substituting \eqref{example:54_phi} and \eqref{example:54_cons} into \eqref{energy}, and neglecting the kinetic energy, we obtain the approximate energy
\begin{equation}\label{example:54_energy}
E_{\mathrm{app}}(a)
=
\frac{\beta}{2}ca^2+\frac{2\lambda}{5}ca^3.
\end{equation}
Minimizing \(E_{\mathrm{app}}(a)\) with respect to \(a\) yields
\[
a= -\frac{5\beta}{6\lambda},
\qquad
E_{\mathrm{app}}^{\min}
=
E_{\mathrm{app}}\!\left(-\frac{5\beta}{6\lambda}\right)
=
\frac{\beta c}{6}\left(\frac{5\beta}{6\lambda}\right)^2.
\]
We then compare these predictions with the ground states. Here \(\phi_g\) denotes the computed ground state obtained on a sufficiently fine mesh. Since the peak is attained at the origin, \(\phi_g(0)\) represents its peak value. The relative errors are defined by
\begin{equation}\label{def:rel_err}
e_a=\frac{|\phi_g(0)-a|}{|\phi_g(0)|},
\qquad
e_E=\frac{|E(\phi_g)-E_{\mathrm{app}}^{\min}|}{|E(\phi_g)|}.
\end{equation}

\begin{example}[Verification of the flat-top approximation]
Under the free-space condition \(V(\mathbf x)\equiv 0\), we fix \(c=1\) and vary \(\beta\) and \(\lambda\) separately to compare the computed ground states with the above heuristic estimates in the limiting regime. 
\end{example}

\begin{table}[htp!]
\centering
\begin{tabular}{|c|ccc|ccc|}
\hline
\(\beta\) & \(\phi_g(0)\) & \(a\) & \(e_a\) & \(E(\phi_g)\) & \(E_{\mathrm{app}}^{\min}\) & \(e_E\) \\
\hline
-10  & 8.81E1 & 8.33E1 & 5.42E-2 & -7.76E3 & -1.16E4 & 4.92E-1 \\
-50  & 4.23E2 & 4.17E2 & 1.58E-2 & -1.34E6 & -1.45E6 & 8.47E-2 \\
-100 & 8.41E2 & 8.33E2 & 9.04E-3 & -1.11E7 & -1.16E7 & 4.42E-2 \\
-250 & 2.09E3 & 2.08E3 & 4.26E-3 & -1.77E8 & -1.81E8 & 1.97E-2 \\
-500 & 4.18E3 & 4.17E3 & 2.41E-3 & -1.43E9 & -1.45E9 & 1.12E-2\\
\hline
\end{tabular}
\caption{Computed and estimated peak values and energies, together with the relative errors \(e_a\) and \(e_E\) defined in \eqref{def:rel_err}, for different values of \(\beta\) under \(V(\mathbf x)\equiv 0\), \(c=1\) and \(\lambda=0.1\).}
\label{tab:beta_estimate}
\end{table}

\begin{table}[htp!]
\centering
\begin{tabular}{|c|ccc|ccc|}
\hline
\(\lambda\) & \(\phi_g(0)\) & \(a\) & \(e_a\) & \(E(\phi_g)\) & \(E_{\mathrm{app}}^{\min}\) & \(e_E\) \\
\hline
0.1   & 8.81E1 & 8.33E1 & 5.42E-2 & -7.76E3 & -1.16E4 & 4.92E-1 \\
0.02  & 4.31E2 & 4.17E2 & 3.38E-2 & -2.36E5 & -2.89E5 & 2.24E-1 \\
0.01  & 8.57E2 & 8.33E2 & 2.73E-2 & -9.92E5 & -1.16E6 & 1.66E-1 \\
0.004 & 2.13E3 & 2.08E3 & 2.04E-2 & -6.49E6 & -7.23E6 & 1.15E-1 \\
0.002 & 4.24E3 & 4.17E3 & 1.64E-2 & -2.66E7 & -2.89E7 & 8.84E-2\\
\hline
\end{tabular}
\caption{Computed and estimated peak values and energies, together with the relative errors \(e_a\) and \(e_E\) defined in \eqref{def:rel_err}, for different values of \(\lambda\) under \(V(\mathbf x)\equiv 0\), \(c=1\) and \(\beta=-10\).}
\label{tab:lambda_estimate}
\end{table}

From Tables~\ref{tab:beta_estimate} and~\ref{tab:lambda_estimate}, we observe that the relative errors in both the peak value and the energy decrease as \(|\beta|\) increases or \(\lambda\) decreases. This indicates that the flat-top approximation becomes more accurate in the limiting regime of strong attraction or weak higher-order repulsion, and thus provides a leading-order description of the droplet regime.

\subsection{Two- and three-dimensional computations}

We next present several two- and three-dimensional examples based on the finite element discretization described in Section~\ref{sec:scheme}. Unlike the symmetric cases discussed above, these computations are carried out directly in higher dimensions and allow for more general spatial configurations. They further demonstrate the effectiveness of the proposed method for localized ground states with sharp transition layers.

\begin{example}[Two-dimensional free-space case] \label{example:2D}
In this example, we consider the two-dimensional free-space case with \(c=20\), \(\beta=-10\), \(\lambda=0.1\), and \(V(x,y)\equiv 0\).

The computational domain is \(D=[-1,1]\times[-1,1]\). The computation is performed with \(\tau=10^{-2}\) on an initial coarse mesh with 50 divisions on each boundary. Starting from a normalized Gaussian initial guess, the computation terminates when
\(
\|\phi^{n+1}_h-\phi^n_h\|_2<10^{-6}.
\)

Figure~\ref{fig:2d_none} shows the computed ground-state density together with the corresponding adaptive mesh. The mesh adaptation is carried out in FreeFEM according to the local variation of the numerical solution \cite{hecht2005freefem++}.  As a result, the mesh is refined near the narrow transition layer, where the solution varies rapidly, and remains relatively coarse in the inner high-density region and the outer low-density region. This demonstrates that the adaptive finite element method is effective for resolving localized structures in two dimensions.
\end{example}

\begin{figure}[htp!]
    \centering
    \includegraphics[width=0.99\linewidth]{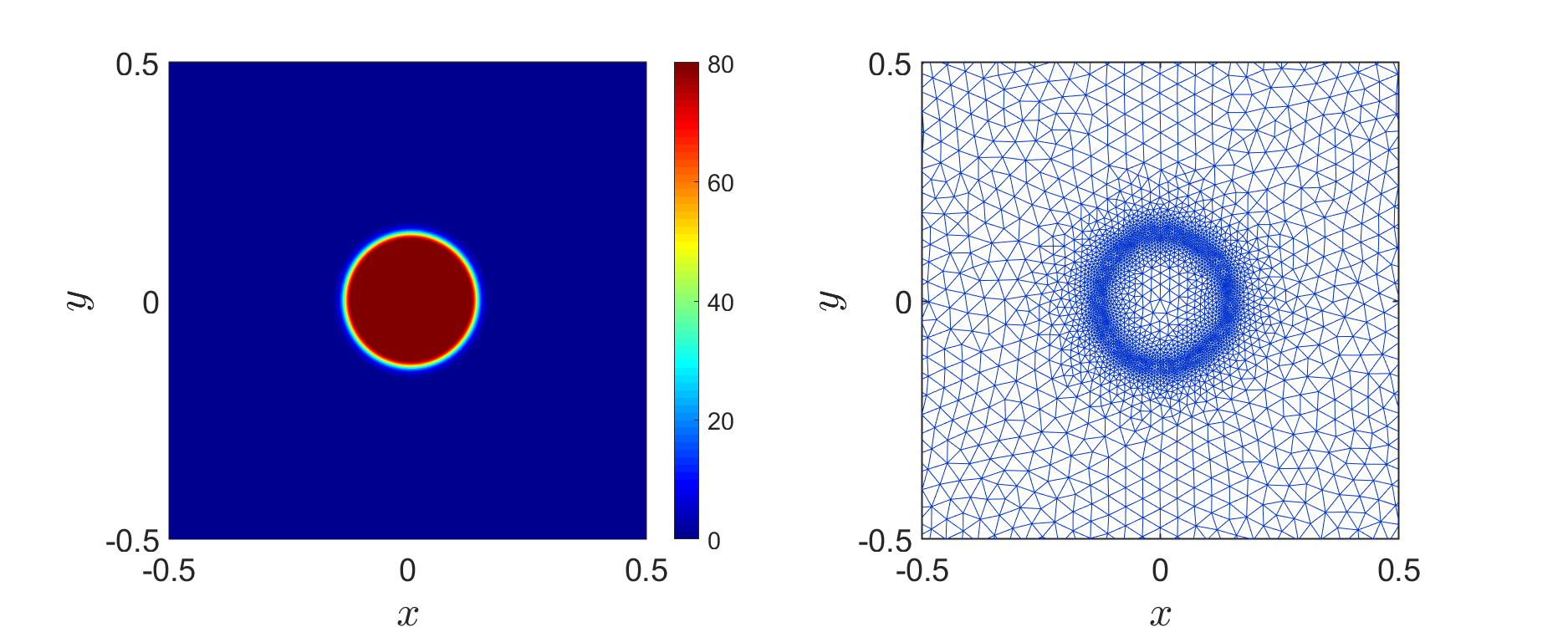}
     \caption{Two-dimensional computation without an external potential. The left panel shows the ground-state density, while the right panel shows the corresponding adaptive mesh.}
    \label{fig:2d_none}
\end{figure}

\begin{example}[Two-dimensional optical lattice potential]
In this example, we consider the two-dimensional case with the optical lattice potential
\[
V(x,y)=V_0(\cos(5\pi x)+\cos(5\pi y)),
\]
with \(V_0=10^3\) and \(V_0=3\times 10^3\).
The computational domain, time step, initial guess, and all other parameters are kept the same as in Example~\ref{example:2D}.

Figure~\ref{fig:2d_lattic} shows the computed ground states for these two values of \(V_0\). When \(V_0\) is relatively small, the solution exhibits several connected peaks, and the density in the regions between them remains non-negligible. As \(V_0\) increases, the peaks become more separated and the density between neighboring peaks decreases significantly, indicating stronger localization induced by the optical lattice potential.

\begin{figure}[htp!]
    \centering
    \includegraphics[width=0.99\linewidth]{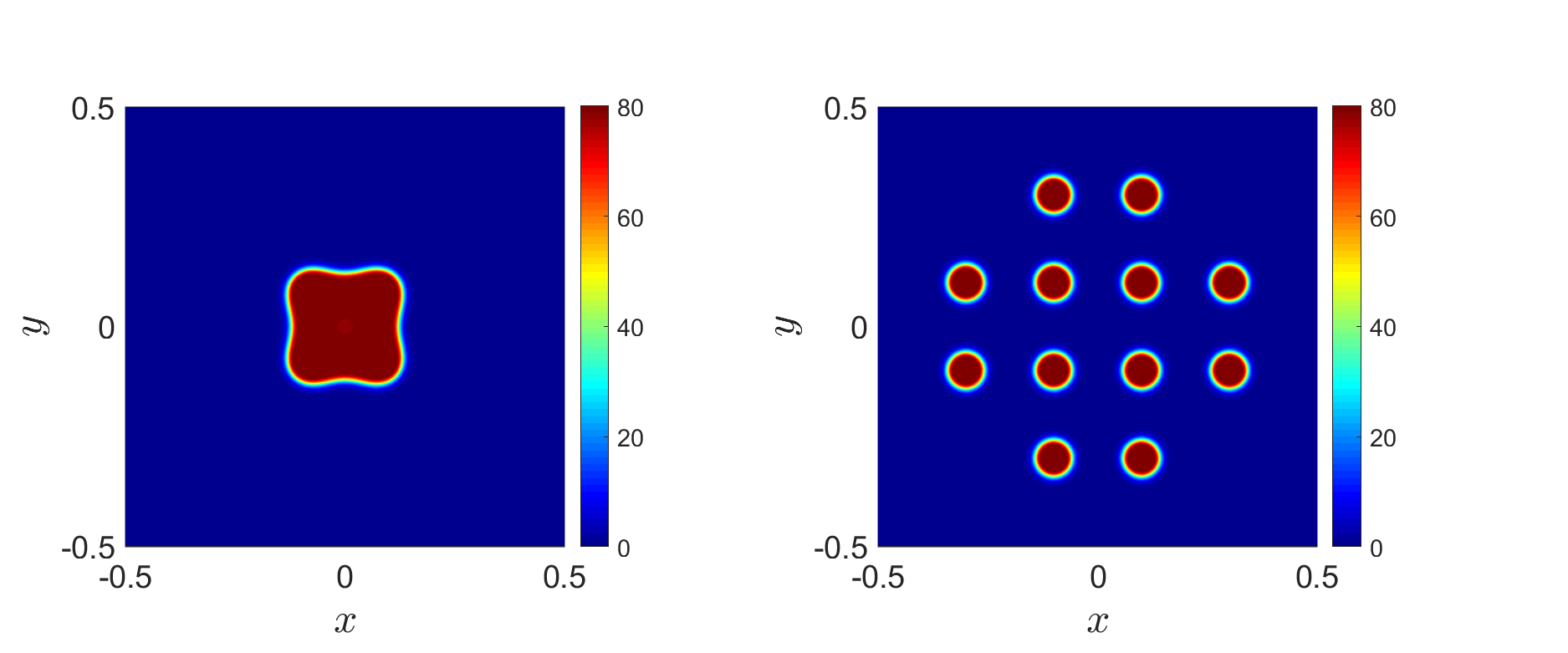}
  \caption{Ground-state densities under the optical lattice potential \(V(x,y)=V_0(\cos(5\pi x)+\cos(5\pi y))\) with \(V_0=10^3\) and \(V_0=3\times 10^3\). }
    \label{fig:2d_lattic}
\end{figure}
\end{example}

\begin{example}[Three-dimensional cases]
In this example, we consider the three-dimensional case with \(c=20\), \(\beta=-10\), and \(\lambda=0.1\), under two different external potentials, namely the free-space case \(V(x,y,z)\equiv 0\) and the anisotropic harmonic-potential case
\(
V(x,y,z)=\frac12(x^2+y^2+10^4z^2).
\)

The computational domain is \(D=[-1,1]\times[-1,1]\times[-1,1]\). The computation is performed with \(\tau=10^{-2}\), starting from a normalized Gaussian initial guess, and terminates when
\(
\|\phi_h^{n+1}-\phi_h^n\|_2<10^{-6}.
\)

Figure~\ref{fig:3d} shows three-dimensional visualizations of the computed ground states, using a surrounding surface together with three planar cross-sections. In the free-space case, the profile remains nearly isotropic and is very close to a spherical configuration. In the anisotropic harmonic-potential case, the solution is strongly compressed in the \(z\)-direction, while remaining more extended in the \(x\)-\(y\) plane. In both cases, the density remains localized and exhibits a thin transition layer near the boundary of the high-density core. These results show that the proposed method can effectively capture localized ground states in three dimensions and their deformation under anisotropic confinement.

\end{example}

\begin{figure}[htp!]
    \centering
    \begin{subfigure}[b]{0.45\textwidth}
        \centering
        \includegraphics[width=\linewidth]{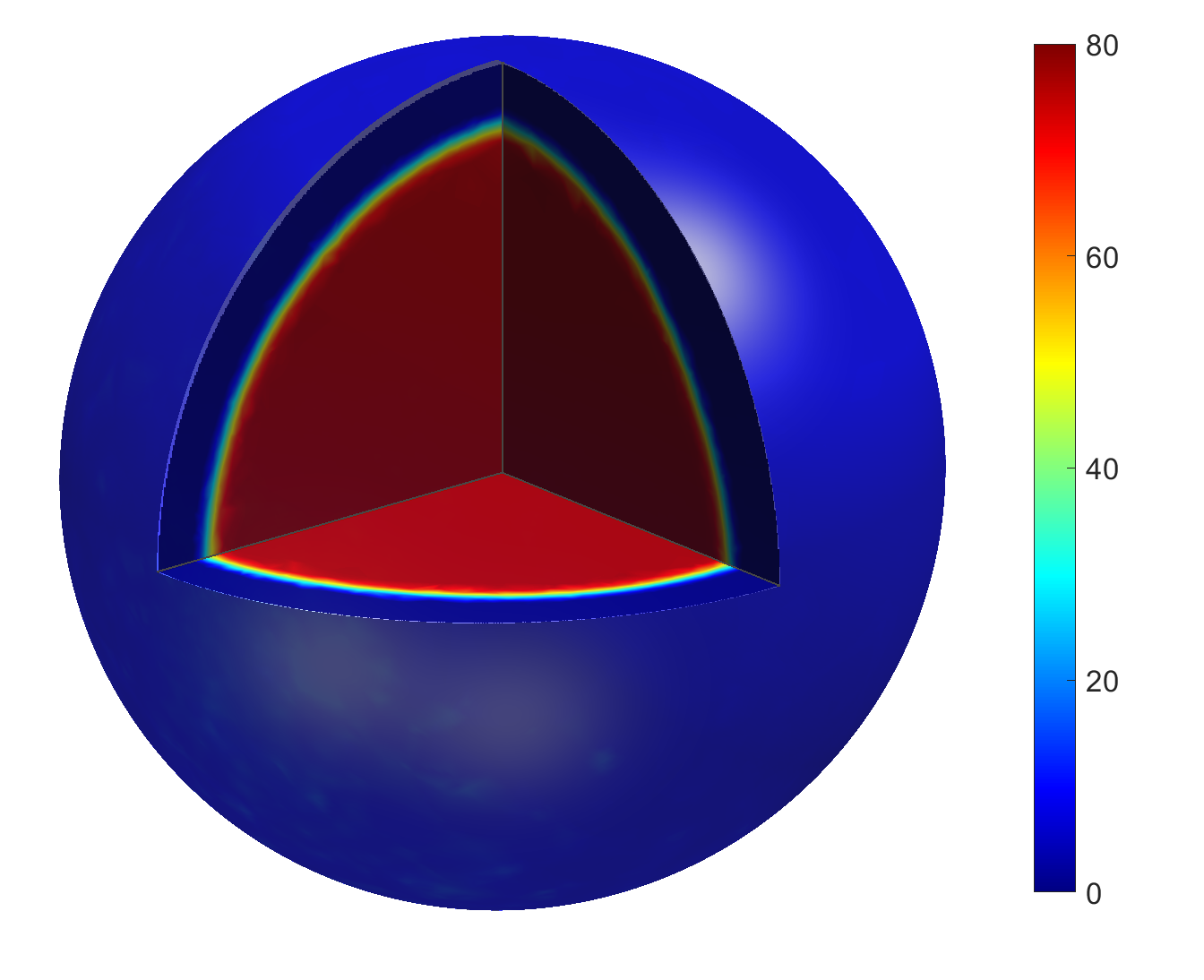}
        \label{subfig:3d1}
    \end{subfigure}
    \hspace{1em}
    \begin{subfigure}[b]{0.45\textwidth}
        \centering
        \includegraphics[width=\linewidth]{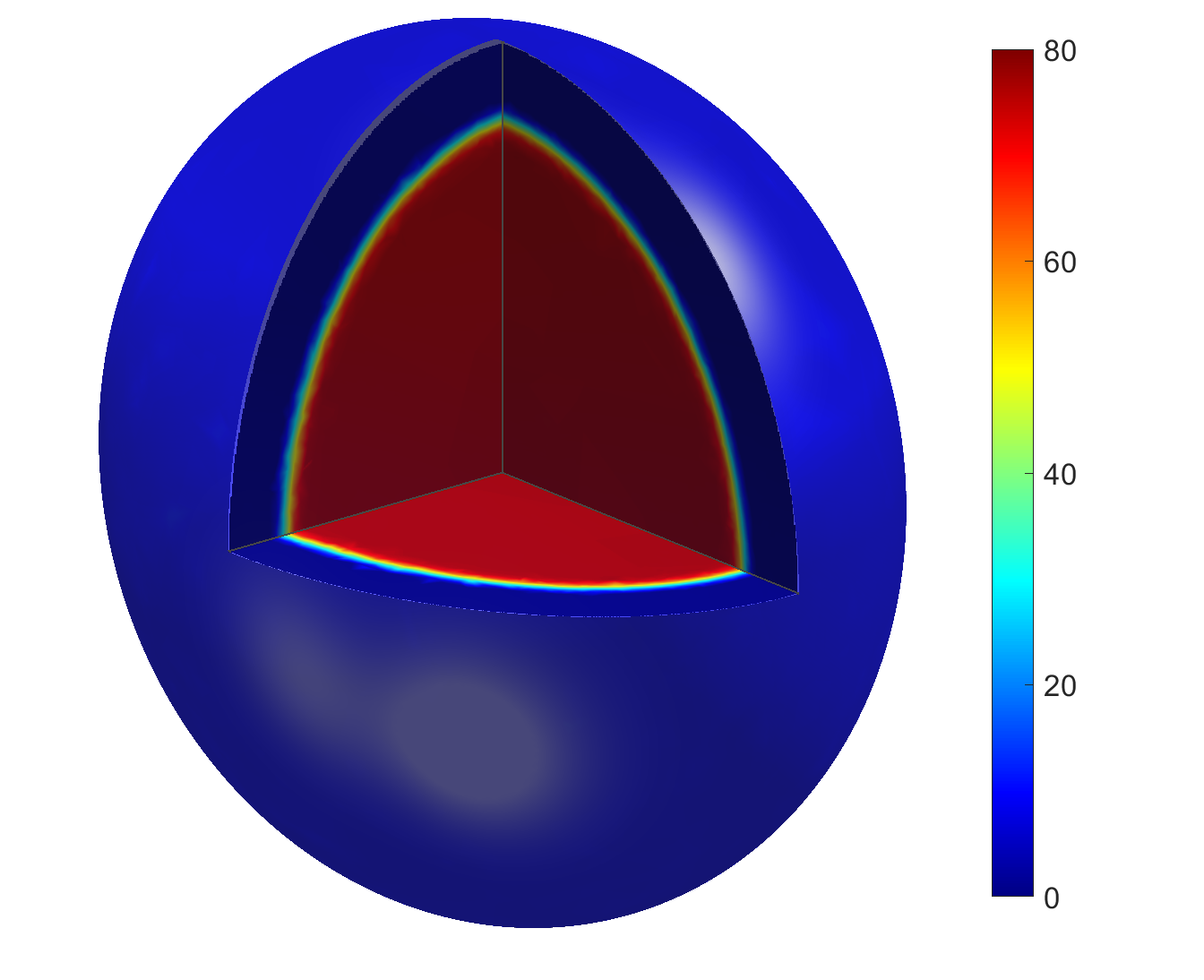}
        \label{subfig:3d2}
    \end{subfigure}
    \caption{Three-dimensional ground-state densities for \(c=20\), \(\beta=-10\), and \(\lambda=0.1\), shown by a surrounding surface together with planar cross-sections. The left panel corresponds to the free-space case \(V(x,y,z)\equiv 0\), and the right panel corresponds to the anisotropic harmonic-potential case \(V(x,y,z)=\frac12(x^2+y^2+10^4z^2)\).}
    \label{fig:3d}
\end{figure}


\section{Conclusion}
\label{sec:conclusion}

In this work, we investigated the ground states of the extended Gross--Pitaevskii (eGP) equation with the Lee--Huang--Yang (LHY) correction from both theoretical and numerical perspectives. Starting from the three-dimensional model, we derived reduced one- and two-dimensional equations through nondimensionalization and dimensional reduction, and wrote the resulting problems in a unified form. For this unified eGP model, we established existence and nonexistence results for ground states in different spatial dimensions, both in the free-space case and in the presence of a confining external potential.

For the numerical computation of ground states, we proposed a normalized gradient flow method with a Lagrange multiplier and combined it with a finite element discretization in space. For symmetric settings, the problem was further reduced to a one-dimensional formulation, which allowed a systematic study of parameter effects. The numerical results showed how the mass, interaction strength, and LHY coefficient influence the ground-state profiles, and revealed different regimes in the parameter plane, including no-ground-state, soliton-like, and droplet-like regions. In the droplet regime, we further introduced a simple flat-top approximation and verified its accuracy in the limiting parameter range. We also presented two- and three-dimensional computations to illustrate more general localized structures and their deformation under external potentials.

Overall, the present work provides both analytical results and effective numerical methods for the study of ground states in eGP models with the LHY correction.

\appendix
\renewcommand{\thesection}{Appendix \Alph{section}}
\renewcommand{\theequation}{\Alph{section}.\arabic{equation}}

\titleformat{\section}
  {\normalfont\large\bfseries}
  {\thesection}{1em}{}

\section{Finite difference discretization in symmetric settings}\label{app:sym}

When the external potential \(V(\mathbf{x})\) is even in one dimension, radially symmetric in two dimensions, or spherically symmetric in three dimensions, the corresponding symmetric solution depends only on the variable \(r\in[0,\infty)\), where \(r=|x|\) for \(d=1\) and \(r=|\mathbf{x}|\) for \(d=2,3\). In this situation, the eGP equation \eqref{eGPE} becomes
\begin{equation}
i\partial_t\phi(r,t)
=
-\frac{1}{2r^{d-1}}\frac{\partial}{\partial r}
\Bigl(r^{d-1}\frac{\partial \phi}{\partial r}\Bigr)
+
\left(V(r)+\beta|\phi|^2+\lambda|\phi|^3\right)\phi,
\qquad r\in(0,\infty).
\label{eq:radial_egpe}
\end{equation}
This equation is supplemented with the boundary conditions
\begin{equation}
\partial_r\phi(0,t)=0,
\qquad
\phi(r,t)\to 0
\quad \text{as } r\to\infty.
\label{eq:radial_bc}
\end{equation}

To compute the ground state, we truncate the semi-infinite interval to \([0,R]\) with sufficiently large \(R\) and discretize it by a midpoint grid with mesh size \(\Delta r=R/M\). Let
\[
r_j=j\Delta r,\qquad
r_{j+\frac12}=\Bigl(j+\frac12\Bigr)\Delta r,
\qquad j=0,1,\dots,M.
\]
Denote by \(\phi_{j+\frac12}^n\) the approximation to \(\phi(r_{j+\frac12},t_n)\). The corresponding discrete operator is defined by
\begin{equation}
\delta_{r,d}^2 \phi_{j+\frac12}
=
\frac{1}{(\Delta r)^2\, r_{j+\frac12}^{d-1}}
\left(
r_{j+1}^{d-1}\phi_{j+\frac32}
-
\bigl(r_{j+1}^{d-1}+r_j^{d-1}\bigr)\phi_{j+\frac12}
+
r_j^{d-1}\phi_{j-\frac12}
\right),
\qquad j=0,\dots,M-1.
\label{eq:radial_laplacian}
\end{equation}

Replacing the Laplacian in \eqref{eq:semi_discrete_scheme} by \(\delta_{r,d}^2\), we obtain the finite difference scheme
\begin{equation}
\frac{\tilde \phi_{j+\frac12}^{n+1}-\phi_{j+\frac12}^n}{\tau}
=
\Bigl(
\frac12 \delta_{r,d}^2
-
V(r_{j+\frac12})
-
\lambda |\phi_{j+\frac12}^n|^3
\Bigr)\tilde \phi_{j+\frac12}^{n+1} +
\mu^n  \phi_{j+\frac12}^{*}
-
\beta |\phi_{j+\frac12}^n|^2\phi_{j+\frac12}^{**},
\, j=0,\dots,M-1,
\label{eq:radial_scheme}
\end{equation}
where \(\mu^n\),  \(\phi_{j+\frac12}^{*}\) and \(\phi_{j+\frac12}^{**}\) are defined in a similar way as in \eqref{eq:mu_n_fem} and \eqref{eq:phi_star_fem}.  The boundary conditions are discretized as
\begin{equation}
\tilde\phi_{-\frac12}^{n+1}=\tilde\phi_{\frac12}^{n+1},
\qquad
\tilde \phi_{M+\frac12}^{n+1}=0.
\label{eq:radial_scheme_bc}
\end{equation}
The intermediate solution is then normalized by
\begin{equation}
\phi_{j+\frac12}^{n+1}
=
c\,\frac{\tilde \phi_{j+\frac12}^{n+1}}{\|\tilde \phi^{n+1}\|_{2}},
\label{eq:radial_normalization}
\end{equation}
where the discrete \(L^2\)-norm is defined as 
\begin{equation}
\|\tilde \phi^{n+1}\|_{2}
=
\biggl(
\omega(d)\Delta r
\sum_{j=0}^{M-1}
\left|\tilde \phi_{j+\frac12}^{n+1}\right|^2
\bigl(r_{j+\frac12}\bigr)^{d-1}
\biggr)^{1/2},
\qquad
\omega(1)=2,\ \omega(2)=2\pi,\ \omega(3)=4\pi.
\label{eq:radial_norm}
\end{equation}

\bigskip
\noindent{\large\bf Acknowledgements}  \ 
X. Ruan acknowledges support  from the National Natural Science Foundation of China under grant  12201436.
W. Huang acknowledges support  from the National Natural Science Foundation of China under grant  12001034.

\providecommand{\bysame}{\leavevmode\hbox to3em{\hrulefill}\thinspace}
\providecommand{\MR}{\relax\ifhmode\unskip\space\fi MR }
\providecommand{\MRhref}[2]{%
  \href{http://www.ams.org/mathscinet-getitem?mr=#1}{#2}
}
\providecommand{\href}[2]{#2}

\end{document}